\documentclass[conference]{IEEEtran}
\usepackage{breakurl}             
\usepackage{underscore}           

\usepackage[final]{commenting}

\newif\ifreview
\reviewfalse

\usepackage[noadjust]{cite}

\usepackage{graphicx}
\usepackage{framed}

\usepackage{savesym}
\savesymbol{comment}

\usepackage{hyperref}

\usepackage[normalem]{ulem}

\usepackage{amsmath,amsfonts,url,amsthm,stmaryrd,bussproofs,prftree}

\usepackage{listings}
\newcommand{\dz}{Datalog$_\textbf{Z}$}
\newcommand{\Swouts}{\sigf}
\newcommand{\Xwouts}{\Xi\setminus\Sigma}

\newcommand{\Z}{\mathbb{Z}}
\newcommand{\WW}{W}
\newcommand{\Nat}{\mathbb{N}}

\newcommand{\eqS}{=_S}

\newcommand{\sem}[1]{\llbracket#1\rrbracket}
\newcommand{\order}{\textrm{order}}

\newcommand{\level}{\mathrm{level}}
\newcommand{\boolsort}{o}

\newcommand{\lr}{\lesssim}
\newcommand{\glr}{\gtrsim}
\newcommand{\fspace}[2]{[{#1}\to{#2}]}
\newcommand{\firstdef}[1]{\emph{#1}}
\newcommand{\natty}{\mathsf{Nat}}

\newcommand{\tyzero}{\zeta}
\newcommand{\tyint}{\xi}
\newcommand{\tynint}{\nu}
\newcommand{\bnsem}[1]{\Hf_n\sem{#1}}

\newcommand{\confused}{level-entwined\ }

\newcommand{\oconfused}{entwined\ }

\newcommand{\confusionf}[1]{\langle #1 \rangle^\ell}
\newcommand{\oconfusionf}[1]{\langle #1 \rangle}
\newcommand{\assigned}[2]{#1^{#2}}
\newcommand\init{\mathit{Init}}

\declareauthor{tcb}{tcb}{olive}
\declareauthor{sr}{sr}{teal}
\declareauthor{lo}{lo}{red}
\authorcommand{tcb}{comment}
\authorcommand{sr}{comment}
\authorcommand{lo}{comment}
\declareauthor{dw}{dw}{blue}
\authorcommand{dw}{comment}

\newcommand{\cref}[1]{\ref{#1}}

\newcommand\defn[1]{\textbf{\em #1}}

\usepackage[normalem]{ulem} 
\newcommand\defeq{\coloneqq}
\renewcommand\phi{\varphi}
\newcommand\set[1]{\{{#1}\}}

\usepackage[british,UKenglish]{babel}
\usepackage[T1]{fontenc}
\usepackage{mathtools}
\usepackage{amssymb}
\usepackage{thmtools,thm-restate}
\usepackage{tabularx}
\usepackage{proof}
\usepackage{xcolor}
\usepackage{paralist}
\definecolor{oxblue}{RGB}{0,33,71}
\definecolor{oxlightblue}{RGB}{0,72,205}
\definecolor{oxlightblue2}{RGB}{161,196,208}
\definecolor{oxgold}{HTML}{a0630a
}
\definecolor{oxyellow}{RGB}{243,222,116}
\definecolor{oxbrown}{HTML}{88562e}
\definecolor{oxgreen}{HTML}{3e6111}
\definecolor{oxgreen2}{RGB}{105,145,59}
\definecolor{oxgreen3}{RGB}{185,207,150}

\definecolor{oxred}{HTML}{ac0b1d}
\definecolor{oxred2}{RGB}{190,15,52}
\definecolor{oxpink}{RGB}{235,196,203}
\definecolor{oxorange}{RGB}{207,122,48}

\definecolor{oxgrey}{RGB}{167,157,150}

\usepackage{listings}
\usepackage{stackengine}

\newcommand\sinti[2]{{#1}\llbracket#2\rrbracket}

\newcommand{\from}{:}
\newcommand{\As}{\mathcal A}
\newcommand{\Bs}{\mathcal B}

\newcommand{\Hf}{\mathcal H}
\newcommand{\Sf}{\mathcal S}

\newcommand\nat{\mathbb N}
\newcommand\bool{\mathbb B}
\newcommand\nin{\not\in}

\newcommand\arel{\precsim}

\DeclareMathOperator{\dom}{dom}

\DeclareMathOperator{\andf}{and}
\DeclareMathOperator{\orf}{or}

\DeclareMathOperator{\LIA}{LIA}

\DeclareMathOperator{\FO}{{FO}}

\stackMath

\newcommand{\Res}{\Rightarrow_{\Resh}}

\newcommand\btypes{\mathfrak I}

\newcommand{\stkout}[1]{\ifmmode\text{\sout{\ensuremath{#1}}}\else\sout{#1}\fi}

\newcommand{\tden}[1]{\llbracket #1\rrbracket}

\newcommand{\sigf}{\Sigma_{\mathrm{fg}}}
\newcommand{\sigb}{\Sigma_{\mathrm{bg}}}

\usepackage[inline]{enumitem}
\setlist[enumerate,1]{label={(\roman*)}}
\usepackage{bussproofs}

\usepackage[caption=false]{subfig}

\newlist{thmlist}{enumerate}{1}
\setlist[thmlist]{label=\textup{(\roman{thmlisti})},ref={(\roman{thmlisti})},noitemsep}

\newlist{complist}{enumerate}{1}
\setlist[complist,1]{noitemsep,label=\textbf{(C\arabic*)},leftmargin=2.5\parindent}

\usepackage{makecell} 

\usepackage{pifont}

\makeatletter

\renewcommand{\p@thmlisti}{\perh@ps{\thetheorem}}
\protected\def\perh@ps#1#2{\textup{#1#2}}
\newcommand{\itemrefperh@ps}[2]{\textup{#2}}
\newcommand{\itemref}[1]{\begingroup\let\perh@ps\itemrefperh@ps Part~\ref{#1}\endgroup}
\newcommand{\itemrefs}[2]{\begingroup\let\perh@ps\itemrefperh@ps Parts~\ref{#1} and~\ref{#2}\endgroup}

\renewcommand{\Res}{\vdash_\As}

\newcommand{\Set}{\Gamma}

\renewcommand{\Hf}{\mathcal F}

\declaretheorem[    name=Theorem,    Refname={Theorem,Theorems}, numberwithin=section]{theorem}
\declaretheorem[    name=Lemma,    Refname={Lemma,Lemmas},    sibling=theorem]{lemma}
\declaretheorem[    name=Proposition,    Refname={Prop.,Propositions},    sibling=theorem]{proposition}
\declaretheorem[    name=Corollary,    Refname={Corollary,Corollaries},    sibling=theorem]{corollary}

\declaretheorem[numbered=yes,
    name=Assumption,
    Refname={Assumption,Assumptions}]{assumption}
\theoremstyle{definition}
\declaretheorem[    name=Definition,    Refname={Definition,Definitions},    sibling=theorem]{definition}
\declaretheorem[    name=Example,    Refname={Example,Examples},    sibling=theorem]{example}
\theoremstyle{remark}
\declaretheorem[    name=Remark,    Refname={Remark,Remarks},    sibling=theorem]{remark}
    \declaretheoremstyle[    spaceabove=-4pt,     spacebelow=6pt,     headfont=\it,     bodyfont = \normalfont,    postheadspace=1em,     qed=$\blacksquare$,     headpunct={.}]{myproofstyle} 

\usepackage{hyperref}
\usepackage[capitalize]{cleveref}

\Crefname{theorem}{Thm.}{Theorems}
\Crefname{corollary}{Cor.}{Corollary}
\crefname{proposition}{Prop.}{Propositions}
\Crefname{claim}{Claim}{Claims}
\Crefname{definition}{Def.}{Definitions}
\Crefname{fact}{Fact}{Facts}
\Crefname{conj}{Conjecture}{Conjectures}
\Crefname{example}{Ex.}{Examples}
\Crefname{example}{Rem.}{Remarks}
\Crefname{convention}{Convention}{Conventions}
\Crefname{lemma}{Lem.}{Lemmas}
\Crefname{assumption}{Assumption}{Assumptions}
\Crefname{section}{Sec.}{Sec.}
\Crefname{appendix}{App.}{App.}

\addtotheorempostheadhook[theorem]{\crefalias{thmlisti}{theorem}}
\addtotheorempostheadhook[lemma]{\crefalias{thmlisti}{lemma}}
\addtotheorempostheadhook[proposition]{\crefalias{thmlisti}{proposition}}
\addtotheorempostheadhook[corollary]{\crefalias{thmlisti}{corollary}}
\addtotheorempostheadhook[claim]{\crefalias{thmlisti}{claim}}
\addtotheorempostheadhook[fact]{\crefalias{thmlisti}{fact}}
\addtotheorempostheadhook[conjecture]{\crefalias{thmlisti}{conjecture}}
\addtotheorempostheadhook[example]{\crefalias{thmlisti}{example}}
\addtotheorempostheadhook[definition]{\crefalias{thmlisti}{definition}}
\addtotheorempostheadhook[remark]{\crefalias{thmlisti}{remark}}

\newcommand{\inferbintocl}[5]{\begin{minipage}{12ex}{\bfseries #1}\end{minipage} {
    \begin{minipage}{8ex}{$\infer{#4}{#2 &&& #3}$}\end{minipage}
    \par\noindent #5}}

\newcommand{\inferuntocl}[4]{\begin{minipage}{12ex}{\bfseries #1}\end{minipage} {
    \begin{minipage}{8ex}{$\infer{#3}{#2}$}
    \end{minipage}
    \par\noindent
      #4
    }}

\newcommand{\subparagraph}[1]{\paragraph{#1}}

\usepackage{fancybox}
\newenvironment{fminipage}[1]%
    {\begin{Sbox}\begin{minipage}{#1}}%
  {\end{minipage}\end{Sbox}\fbox{\TheSbox}}

\renewcommand\impliedby{\leftarrow}

\usepackage{calc}
\newlength\dlf
\newcommand\alignedbox[3][black!10]{
  &
  \begingroup
  \settowidth\dlf{$\displaystyle #2$}
  \addtolength\dlf{\fboxsep+\fboxrule}
  \hspace{-\dlf}
  \colorbox{#1}{$\displaystyle #2 #3$}
  \endgroup
}

\begin{document}

\title{Initial limit Datalog: a new extensible class of decidable constrained Horn clauses}


\author{Toby Cathcart~Burn \and Luke Ong \and Steven Ramsay \and Dominik Wagner}

\maketitle

\begin{abstract}
  We present \emph{initial limit Datalog}, a new extensible class of constrained Horn clauses for which the satisfiability problem is decidable.
  The class may be viewed as a generalisation to higher-order logic (with a simple restriction on types) of the first-order language \emph{limit \dz{}} (a fragment of Datalog modulo linear integer arithmetic), but can be instantiated with any suitable background theory.
  For example, the fragment is decidable over any countable well-quasi-order with a decidable first-order theory, such as natural number vectors under componentwise linear arithmetic, and words of a bounded, context-free language ordered by the subword relation.
  Formulas of initial limit Datalog have the property that, under some assumptions on the background theory, their satisfiability can be witnessed by a new kind of term model {which} we call \emph{entwined structures}.
  Whilst the set of all models is typically uncountable, the set of all entwined structures is recursively enumerable, and model checking is decidable.
\end{abstract}

\IEEEpeerreviewmaketitle

\section{Introduction}
\label{sec:intro}

\emph{Constrained Horn Clauses} (CHCs) are a class of formulas that have been found to be especially suitable for tasks in automated reasoning.
They are the language of constraint logic programming \cite{JAFFAR1994503}. 
More recently, there has been a concerted effort to exploit the class as a programming-language independent basis for automatic program verification \cite{DBLP:conf/birthday/BjornerGMR15,popl18}.

CHCs are a liberalisation of the class of Horn formulas in which, additionally, clauses may contain \emph{constraints} drawn from a specified first-order \emph{background theory}\footnote{Note: in this work we will assume the background theory has a fixed interpretation, as is common in the satisfiability-modulo-theories literature.}.
This extension preserves many of the good properties of the Horn format, such as the existence of canonical models and the sufficiency of SLD-style derivations, whilst allowing for the expression of domain-specific knowledge in the form of assertions from the background \changed[tcb]{theory}.

Unfortunately, this pleasing combination of expressivity and semantic characterisation comes with an algorithmic cost. 
In general, decidability of the satisfiability problem for a class of CHC depends on the choice of background theory, and for many theories that are typical in automated reasoning (e.g. because they are decidable), the class of CHC is undecidable.
For example, \cite{DBLP:journals/amai/CoxMT92} shows that not only is CHC over linear integer arithmetic undecidable \cite{D72}, but so too CHC over complex, real or rational linear arithmetic.
On the other hand, it is easy to see that CHC over the theory of equality on a finite set has decidable satisfiability.

Since the most promising applications concern theories of infinite structures, it becomes important to identify restrictions on the format that both preserve its essential character and yet guarantee decidability.  
In \cite{DBLP:journals/amai/CoxMT92}, a catalogue of (sub-recursive) complexity results are derived concerning  limitations placed on the use of variables within clauses and the nature of parameter passing.  

An alternative approach, and the starting point for the work in this paper, is the \emph{limit} restriction of the language \emph{limit \dz{}}, which was proposed in \cite{DBLP:conf/ijcai/KaminskiGKMH17} as a foundation for declarative data analysis.
Limit \dz{} can be viewed as a language of first-order CHCs over the theory of linear integer arithmetic, but with the following proviso:
predicates in limit \dz{} (called \emph{limit predicates}) are restricted so as to capture only the minimum (or maximum) numeric values in their unique integer parameter.
This restriction ensures that the satisfiability problem is decidable for this class of CHC,
whilst remaining expressive enough to describe important problems in data analysis  (in particular, one may still describe certain kinds of recursively defined predicates over the integers).

One way to implement the limit predicate restriction is to require that all predicates with an integer parameter are either upwards 
or downwards 
closed with respect to that parameter.  We enforce this by the highlighted clauses in the examples below, taken from \cite{DBLP:conf/ijcai/KaminskiGKMH17}.  

The background theory of these examples is the combination of linear integer arithmetic and the theory of equality over a finite set.  
We assume that the elements of this set can be arranged into a linear order $y_1,y_2,\ldots,y_k$ (which will differ from example to example), described by two constraint formulas (i.e. of the background theory) that we will abbreviate $\mathsf{FIRST}(y_1)$ and $\mathsf{NEXT}( y_n, y_{n+1})$.
\begin{example}[Social networking]
\label{eg:fo-tweet}
In this first example, the finite set describes people who tweet and follow each other's tweets.
Let us suppose we have a constraint formula\footnote{One can also think more specifically of an intensional database predicate.} (i.e. of the background theory) abbreviated by $\mathsf{TH}(x, m)$, indicating the \emph{retweet threshold}.  
That is, asserting that a person $x$ will tweet a (hypothetical) message if at least $m$ of those they follow also tweet it.
Suppose we have a constraint formula $\mathsf{FOLLOWS}(x, y)$, describing when one individual $x$ follows the tweets of another $y$. 
The following clauses constrain a proposition $\mathrm{Tw}\, x$ so that it holds if $x$ tweeted.
The proposition $\mathrm{Nt}\, x\, y\, m$ holds if, out of the people at or before $y$ (according to the ordering on the set of people), at least $m$ people that $x$ follows tweeted.

\begin{align*}
	\mathrm{Nt}\, x\, y\, 0 &\impliedby \mathsf{FIRST}(y)\\
	\mathrm{Nt}\, x\, y\, 1 &\impliedby \mathsf{FOLLOWS}(x, y) \land \mathsf{FIRST}(y) \land \mathrm{Tw}\, y\\
	\mathrm{Nt}\, x\, y\, m &\impliedby \mathrm{Nt}\, x\, y'\, m \land \mathsf{NEXT}(y', y)\\
	\mathrm{Nt}\, x\, y\, (m+1) &\impliedby \mathrm{Nt}\, x\, y'\, m \land \mathsf{FOLLOWS}(x, y)\\ 
	&\qquad \quad \land \mathsf{NEXT}(y', y) \land \mathrm{Tw}\, y\\
	\alignedbox{\mathrm{Nt} \; x \; y \; m}{\impliedby m\le n\land \mathrm{Nt} \; x \; y \; n} \\
	\mathrm{Tw}\, x &\impliedby \mathsf{TH}(x, m) \land \mathrm{Nt}\, x\, y\, n \land m\le n
\end{align*}
\end{example}

\begin{example}[Path counting]
\label{eg:fo-paths}
In this second example, the finite set describes the vertices of a directed acyclic graph and the clauses can be used to reason about the number of paths between two nodes. We assume a constraint formula $\mathsf{EDGE}(x, y)$ indicating that there is an edge from $x$ to $y$.
\begin{align*}
	\mathrm{Np'}\, x\, y\, z\, 0 &\impliedby \mathsf{FIRST}(z)\\
	\mathrm{Np'}\, x\, y\, z\, m &\impliedby \mathsf{FIRST}(z) \land \mathrm{Np}\, z\, y\, m \land \mathsf{EDGE}(x, z)\\
	\mathrm{Np'}\, x\, y\, z\, m &\impliedby \mathsf{NEXT}(z', z) \land \mathrm{Np'}\, x\, y\, z'\, m\\
	\mathrm{Np'}\, x\, y\, z\, (m+n) &\impliedby \mathsf{NEXT}(z', z) \land \mathrm{Np'}\, x\, y\, z'\, m \\
	& \qquad \quad \land \mathsf{EDGE}(x, z) \land \mathrm{Np}\, z\, y\, n\\
	\alignedbox{\mathrm{Np'}\, x\, y\, z\, m}{\impliedby m\le n\land \mathrm{Np'}\, x\, y\, z\, n}\\
	\mathrm{Np}\, x\, y\, m &\impliedby \mathrm{Np'}\, x\, y\, z\, m\\
	\mathrm{Np}\, x\, x\, 1 &\impliedby \mathsf{true}\\
	\alignedbox{\mathrm{Np}\, x\, y\, m}{\impliedby m\le n\land \mathrm{Np}\, x\, y\, n}
\end{align*}
Here, $\mathrm{Np'}\, x\, y\, z\, m$ holds if there are at least $m$ paths of the form $x, w, \cdots ,y$ where $w$ occurs at or before $z$ according to the linear ordering of nodes.
Finally, $\mathrm{Np}\, x\, y\, n$ holds if there are at least $n$ paths from $x$ to $y$.

\end{example}

\subsection*{Contributions}
In this paper, we introduce a significant \changed[lo]{yet} decidable extension of limit \dz{} which we call \emph{initial limit Datalog}.
Our language encompasses generalisations of the original work \cite{DBLP:conf/ijcai/KaminskiGKMH17} along two dimensions and we define a new class of models:

\noindent
\emph{(i) Parametrisation with respect to a wide range of background theories.}  We give a number of abstract conditions on the character of the background theory which, if satisfied, guarantee decidability of the language (\cref{thm:fold}).  Instances of particular note include all countable well-quasi orders (WQOs) with a decidable first-order theory. This contains, for example, the theory of tuples of naturals under component-wise ordering, allowing the use of predicates with more than one natural number argument.

\noindent
\emph{(ii) (Un)decidability at higher type.} We show that the most natural extension of limit \dz{} to higher-order logic, in which clauses can define predicates of arbitrary higher type, already has undecidable satisfiability (\Cref{thm:undecidability}).
Through a careful analysis of the interaction between the typing discipline and model construction, we design a restriction on the types of predicates (automatically satisfied by all first-order predicates) that we call \emph{initial}.
We show that the resulting language, \emph{initial limit Datalog}, regains decidable satisfiability (\Cref{thm:order-main-sat-result}).

\noindent
\emph{(iii) A recursively enumerable set of candidate models.} The solution space for a given set of clauses is typically uncountable, because predicates are interpreted as subsets of the domain.  A key step in proving our decidability results is to show that, remarkably, one can restrict attention to a recursively enumerable class of candidate models.  To handle the higher-order case, we introduce a new representation which we call \emph{entwined structures}, in which the interpretation of a higher \emph{type} may depend on the interpretation of particular \emph{terms} of lower types.  They have many useful properties, and their conception is sufficiently general that we believe they may be of use for obtaining similar results beyond the scope of this paper.

\subsection*{Initial limit Datalog}
The setting for our language is the fragment of higher-order logic known as \emph{higher-order constrained Horn clauses} (HoCHC) \cite{popl18,DBLP:conf/lics/OngW19}.
\emph{Higher-order} constrained Horn clauses allow for the description of  predicates of higher-types (i.e., whose subjects may themselves be predicates).
Such predicates can be described by clauses built from terms of the simply typed $\lambda$-calculus when equipped with the appropriate logical constants.  
As in \cite{popl18,DBLP:conf/lics/OngW19}, we forgo the use of explicit abstraction to simplify the Horn clause format.

As a first example, we demonstrate in \cref{eg:tweet} and \cref{eg:paths} how the first-order limit \dz{} examples above share a common structure which can be factored out into a higher-order recursion combinator $\mathrm{Iter}$ of the following type:
\[
	\mathrm{Iter} : S \to \Z \to (S \to \Z \to o) \to o
\]
Throughout (the examples of) this paper, we will use $S$ to denote the type of a fixed finite set, $o$ as the type of propositions and, by some abuse, $\Z$ as the type of the integers.
This combinator can be defined as follows:
\begin{align*}
	\mathrm{Iter} \, y \, n \, p &\impliedby \mathsf{FIRST}(y) \land p \, y \, n\\
	\mathrm{Iter} \, y \, n \, p &\impliedby \mathsf{NEXT}(y', y) \land p \, y \,k \land \mathrm{Iter}\, y' \, m \, p \land n=k+m\\
	\alignedbox{\mathrm{Iter} \, y \, n \, p}{\impliedby n\le m \land \mathrm{Iter} \, y \, m \, p}
\end{align*}
The proposition $\mathrm{Iter} \, y \, n \, p$ describes iteration over a generic sequence of data items in $S$ from the first item until item $y$, evaluating the predicate $p : S \to \Z \to o$ on each item and summing the associated integers to $n$.  
As in the first-order case, we must implement the limit predicate restriction, so we include the shaded clause to guarantee the (in this case) downwards closure of its integer argument.

\begin{example}[Refactoring social networking]
	\label{eg:tweet}
	Using $\mathrm{Iter}$, the whole of the social network example \cref{eg:fo-tweet}, in which the data items are users, can be encoded more concisely as:
	\begin{align*}
		\mathrm{Inc?}\, x \, y \, n &\impliedby n = 0 \lor \big(\mathsf{FOLLOWS}(x, y) \land \mathrm{Tw} \, y \land n = 1\big)\\
		\alignedbox{\mathrm{Inc?}\, x \, y \, n}{\impliedby n\le m\land \mathrm{Inc?} \, x \, y \, m}\\
		\mathrm{Tw} \, x &\impliedby \mathsf{TH}(x, m) \land \mathrm{Iter} \, y \, n \, (\mathrm{Inc?}\,x) \land m\le n
	\end{align*}

\end{example}
The predicate $\mathrm{Inc?}$, which satisfies the limit restriction, expresses the domain specific reasoning that happens on each iteration, namely that the number of tweeters will either be increased by 0, or by 1 in case $x$ follows some $y$ who tweets the message.

\begin{example}[Refactoring path counting]\label{eg:paths}
	The path counting example \cref{eg:fo-paths} uses a similar iterative structure.  The whole example can be rewritten as:
	\begin{align*}
		\mathrm{NpExt} \, x \, y \, z \, m &\impliedby \big(\mathrm{Np} \, z \, y\, m \land \mathsf{EDGE}(x, z)\big) \lor m = 0\\
		\alignedbox{\mathrm{NpExt} \, x \, y \, z \, m}{\impliedby m\le n\land \mathrm{NpExt} \, x \, y \, z \, n}\\
		\mathrm{Np} \, x \, y \, m &\impliedby \mathrm{Iter}\, z \, m\, (\mathrm{NpExt}\, x\,y) \lor (x=y \land m = 1)\\
		\alignedbox{\mathrm{Np} \, x \, y \, m}{\impliedby m\le n\land \mathrm{Np} \, x \, y \, n}
	\end{align*}
    In this case, the second and fourth clauses ensure that the respective predicates adhere to the limit restriction.

\end{example}

\begin{example}[Generic query] \label{eg:tweet-follows}
An orthogonal benefit of higher-type predicates is to allow the expression of higher-order properties (e.g. properties of the form \emph{for all relations $r$...}).  
Returning to \cref{eg:fo-tweet}, the follows relation was fixed by some first-order constraint formula (or intensional database predicate) $\mathsf{FOLLOWS}(x,y)$.  
Using predicates of higher type, we can instead parametrise the mutually recursive predicates $\mathrm{Nt}$ and $\mathrm{Tw}$ by an \emph{arbitrary} follows relation $f$ of type $S \to S \to o$:
\begin{align*}
	\mathrm{Nt}\, x\, y\, 0\, f &\impliedby \mathsf{FIRST}(y)\\
	\mathrm{Nt}\, x\, y\, 1\, f &\impliedby f \, x \, y \land \mathsf{FIRST}(y) \land \mathrm{Tw}\, y\, f\\
	\mathrm{Nt}\, x\, y\, m\, f &\impliedby \mathrm{Nt}\, x\, y'\, m \land \mathsf{NEXT}(y', y)\\
	\mathrm{Nt}\, x\, y\, (m+1)\, f &\impliedby \mathrm{Nt}\, x\, y'\, m \land f \, x \, y\\ 
	&\qquad \quad \land \mathsf{NEXT}(y', y) \land \mathrm{Tw}\, y\\
	\alignedbox{\mathrm{Nt} \, x \, y \, m \, f}{\impliedby m\le n\land \mathrm{Nt} \, x \, y \, n \, f} \\
	\mathrm{Tw}\, x\, f &\impliedby \mathsf{TH}(x, m) \land \mathrm{Nt} \, x \, y \, n \, f \land m\le n
\end{align*}
This allows us to check that a property of the system holds \emph{independently} of who follows whom.
For example, according to Kaminski et al.'s formulation, nobody will tweet the message if we fix all retweet thresholds at 1.
To verify this, we set $\mathsf{TH}(x,m)$ to the constraint formula $m = 1$ and decide satisfiability of the clauses extended with the following goal:
\[
	\mathit{false} \impliedby \mathrm{Tw} \, x \, f
\]
From the satisfiability of the clauses, we can deduce that there does not exist a choice of an individual $x$ and a followers relation $f$ for which the message would be tweeted.
\end{example}

\cref{eg:tweet,eg:paths,eg:tweet-follows} are not limit \dz{} problems, but they are problems of our generalisation: \emph{initial limit Datalog}.
As well as admitting the definition of higher-order relations, in place of the  theory of integer linear arithmetic we allow for the theory of any preordered set $W$ satisfying certain conditions.

\medskip

\paragraph*{Initial limit Datalog problem and satisfiability} Henceforth let $\WW$ be a preordered set with a decidable first-order theory, such that every upwards closed subset of $\WW$ is definable in the theory.
We consider relational types generated from $\WW$ and any finite set $S$ (abusing notation by naming the types after their interpretations).

\noindent\begin{fminipage}{.965\columnwidth}
An \textbf{\em initial limit Datalog problem} 
is a (finite) set $\Gamma$ of HoCHC clauses over $\WW$ and $S$ 
 such that for every predicate $X : \rho$ in the signature,
with $\rho = \sigma_1 \to \cdots \to \sigma_k \to \boolsort$ of order $n$ (say):
\begin{compactenum}[(i)]
\item $\rho$ is \defn{initial}, meaning $\sigma_j = \WW$ for at most one $j$, and if there is such a $j$ then for all $i < j$, $\order(\sigma_i) < \order(\sigma_j \to \cdots \to \sigma_k \to \boolsort)$; \changed[lo]{moreover each $\sigma_i$ is $S$, or $W$, or initial.}


\item if $\sigma_j = \WW$ for some $j$, then $\Gamma$ contains the \defn{limit clause}
\[
{X \, \overline{z} \, x \, \overline{z'} \; \impliedby \;
x \le y \land X \, \overline{z} \, y \, \overline{z'}}
\]
(and $\rho$ is called an \defn{active type}).
\end{compactenum}
The \defn{satisfiability problem} for {\em initial limit Datalog} asks: given an initial limit Datalog problem $\Gamma$, is it satisfiable (modulo the theory of $W$ and $S$)?
\end{fminipage}\\[2mm]

\tcb{I've removed the comment about upwards vs downwards closed. Since the examples now match the definition next to them, we don't need to say as much, but it might be worth saying something here in addition to what I've mentioned in \cref{rem:directiondoesntmatter}.}


We show in \Cref{sec:undec-hold} that a na\"{i}ve extension to higher order leads to undecidability, but the forgoing examples and those we will present in the sequel all obey a certain discipline in the way that the background type $\WW$ and higher types interact.
This is captured by the \emph{initial} restriction, (i), which requires that the types of terms that may be captured by a partial application are of strictly lower order than the partial application itself.
It is easy to verify that this condition holds for the type of $\mathrm{Iter}$ and one can also see it in the types of our higher-order generalisation of $\mathrm{Nt}$ and (trivially) $\mathrm{Tw}$:
\begin{align*}
	\mathrm{Tw} &: S \to (S \to S \to o) \to o\\
	\mathrm{Nt} &: S \to S \to \Z \to (S \to S \to o) \to o
\end{align*}

Note that all formulas of limit \dz{} already satisfy requirements (i) and (ii); and $\Z$, under the theory of linear integer arithmetic, is an appropriate instantiation of $\WW$.

Parametrisation of initial limit Datalog by the type $\WW$ allows for a variety of interesting background structures beyond integer linear arithmetic.
For example, any \changed[lo]{countable} well-quasi-ordering with a decidable background theory (which must include constants for each element of the structure) satisfies the requirements on $\WW$, such as:
\begin{asparaenum}[a.]
\item Tuples of natural numbers, under componentwise ordering with the theory of linear arithmetic on components.

\item Words of a bounded, context-free language, under the subword order \cite{DBLP:conf/fossacs/KuskeZ19}.
%
%

\item Basic process algebra under the subword order. BPA is an automatic structure, and so, has a decidable first-order theory.
There are other examples in the same vein, e.g., communicating finite-state machines \cite{DBLP:journals/tcs/FinkelS01}.
\end{asparaenum}





The following example is a higher-order instance of initial limit datalog where the preorder $W$ is the WQO of tuples of natural numbers, with the theory of linear arithmetic on components.
Notice that in this case, there may be multiple consecutive parameters of type $\mathbb N$ in a predicate.

\begin{example}[Integration]
\label{eg:summation}
Monotone decreasing functions $\mathbb{N}\to\mathbb{N}$ can be represented by downwards closed subsets of $\mathbb{N}\times\mathbb{N}$: such a function $f$ is uniquely identified by $\{(x,y) : y<f(x)\}$. 
Higher-order initial limit Datalog allows us to define a predicate which computes integrals\footnote{$\mathrm{Integral}$ can equivalently be typed as $\mathbb{N} \times \mathbb{N} \to (\mathbb{N} \times \mathbb{N} \to o) \to o$.} over such functions.
\begin{align*}
	\mathrm{Integral}  &: \mathbb{N} \to \mathbb{N} \to (\mathbb{N} \to \mathbb{N} \to o) \to o\\
	\mathrm{Integral}  \, \mathit{tot} \, \mathit{bd}  \, f & \impliedby  \mathit{tot}=0\\
	\mathrm{Integral}  \, \mathit{tot} \, \mathit{bd}  \, f  & \impliedby  \mathit{tot}=x+y+1 \land \mathrm{Integral}  \, x \, (\mathit{bd} +1) \, f \\
	& \qquad \quad \land f \, \mathit{bd}  \, y\\
	\alignedbox{\mathrm{Integral}  \, \mathit{tot} \, \mathit{bd}  \, f}{\impliedby  \mathit{tot} \leq s \land \mathit{bd}  \leq c \land \mathrm{Integral}  \, s \, c \, f}\\\\
	\mathrm{Exp} &: \mathbb{N} \to \mathbb{N} \to o\\
	\mathrm{Exp} \, m \, n & \impliedby   m = 0 \land n < 128\\
	\mathrm{Exp} \, m \, n & \impliedby  \mathrm{Exp}\, x \, y \land m = x-1 \land n+n < y\\
	\alignedbox{\mathrm{Exp} \, m \, n}{\impliedby  m \le x  \land n \le y \land \mathrm{Exp}\, x \, y}\\
	\\
	\mathit{false} & \impliedby  \mathrm{Integral}\,255\,0\, \mathrm{Exp}
\end{align*}

In the canonical interpretation, $\mathrm{Exp}$ represents the function defined by $f(m) = \lfloor 2^{7-m} \rfloor$ and $\mathrm{Integral}  \, \mathit{tot} \, \mathit{bd}  \, f$ is true if $\mathit{tot}$ is less than or equal to the integral (infinite sum) of the monotone function represented by $f$ from $\mathit{bd} $ to $\infty$. 
\changed[tcb]{(Thus $\max \set{\mathit{tot} \mid \hbox{`$\mathrm{Integral}  \, \mathit{tot} \, 0  \, \mathrm{Exp}$' holds}} = 255$.)

This example is unsatisfiable (there is no consistent interpretation of $\mathrm{Integral}$ and $\mathrm{Exp}$ where $\mathrm{Integral}\,255\,0\, \mathrm{Exp}$ is false), but if the constant 255 is changed to 256, it becomes satisfiable.}
\end{example}









\subsection*{Entwined structures}

The key innovation of our decidability proof is the construction (given $\Gamma$) of a set of candidate models, called \emph{entwined structures}, which satisfy a number of pleasing properties:
\begin{compactenum}[(P1)]
\item The set of entwined structures is r.e.

\item In each order-$n$ entwined structure, the denotation of each (initial) relational type (that occurs in $\Gamma$) of order less than $n$ is finite.

\item There is an algorithm that checks if a given entwined structure models $\Gamma$.

\item There is an entwined structure that models $\Gamma$ if and only if $\Gamma$ is satisfiable.

\end{compactenum}


Entwined structures are built up by induction on order, via a bootstrapping process.
Their name reflects the interplay between the interpretation of terms and types during this process:
the interpretation of a type of order-$n$ (the set from which the interpretations of order-$n$ predicate symbols are chosen) can only be given once the interpretation of the relevant predicate symbols of lower-order types has already been fixed.
A family of structures $\set{\Bs_n}_{n \in \omega}$, indexed by (order) $n$, is \emph{entwined}, if $\Bs_0$ is the structure on the empty signature; and in each $\Bs_n$:
\begin{itemize}
\item Predicate symbols in $\Bs_{n-1}$ (those of the foreground signature of order $<$ $n$) are interpreted as per $\Bs_{n-1}$.

\item Each predicate of an order-$n$ active type $\rho = \sigma_1 \to \cdots \to \sigma_k \to \boolsort$ is interpreted as a function (in the set-theoretic $[\Bs_{n}\sem{\sigma_1} \to \cdots \to \Bs_n\sem{\sigma_k} \to \mathbb{B}]$) monotone in the $\WW$-typed argument.\tcb{We could say something here like "may be chosen freely"}
\end{itemize}
For types $\rho=\tau\to\sigma$ of order less than $n$, $\Bs_n\sem{\rho}$ is  the full function space $\fspace{\Bs_n\sem{\tau}}{\Bs_n\sem{\sigma}}$ if that is finite, otherwise it is the least collection of relational functions allowing it to support the interpretations of predicates assigned by $\Bs_{n-1}$. This results in something similar to a term model. We cannot use the term model because there can be infinitely many terms and therefore uncountably many interpretations of higher-order predicates, but our decidability proof rests on enumeration.

\tcb{In the proof, we are more explicit about frames vs structures i.e. $\Bs\sem{\rho} = \oconfusionf{\Bs_{n-1}}_n\sem{\rho}$. Should we say something about that here to make it clearer for people looking at the proof?}
\lo{No. We should just discuss the key ideas in the introduction. Besides we are running out of space.}




In an unrestricted setting, it would not make sense to interpret all the order-$(n-1)$ active predicates (i.e.~predicates of active type) before interpreting the order-$n$ predicates, because an order-$(n-1)$ active predicate may be passed an argument involving a predicate of order-$n$.

However, thanks to the initial type restriction, if an order-$n$ term $N$ of an active type has an order-$m$ subterm $M$ with $m> n$, then $M$ is a subterm of some $L$ (another subterm of $N$) of type $\sigma$ (say) whose order is less than $n$.
Since $\Bs_n\sem{\sigma}$ is finite (P2), we don't need to know all possible values of $M$ to know all possible values of $N$.

We show decidability (\cref{thm:order-main-sat-result}) by exhibiting two semi-decision procedures---one for proving the existence of a model, and the other for non-existence---and running them in parallel.
The former semi-decision procedure is an immediate consequence of (P1), (P3) and (P4).
The latter is an application of the semi-decidability of HoCHC unsatisfiability, via a refutationally complete resolution proof system ($\Gamma$ is unsatisfiable if, and only if, there is a resolution proof of $\bot$ from $\Gamma$) \cite{DBLP:conf/lics/OngW19}.

\subsubsection*{Outline}

We begin with some technical preliminaries in \cref{sec:prelims} before introducing (higher-order) limit Datalog in \cref{sec:ldz}.
We give a proof that the first-order fragment has a decidable satisfiability problem and show that satisfiability in general is undecidable.
In \cref{sec:order-initiality} we present initial restriction on types, and prove that the initial limit Datalog satisfiability problem is decidable.
In \cref{sec:examples}, we give examples of how first-order limit Datalog problems can be used with the background theory of tuples of naturals, and other well-quasi orderings (WQOs) with a decidable first-order theory. 
After a review of related work (\cref{sec:related}), we conclude and briefly discuss some further directions.

\section{Technical preliminaries}
\label{sec:prelims}


This section introduces 
a restricted form
of higher-order logic (\cref{sec:relhol}), higher-order constrained
Horn clauses (HoCHCs) (\cref{ch:defHoCHC}) and their proof system (\cref{sec:resolution}).

\subsection{Relational higher-order logic}
\label{sec:relhol}
\subsubsection{Syntax}

For a fixed set $\btypes$ (intuitively the types of individuals), the set of \emph{argument types}, \emph{relational types}, \emph{1st-order types} and \emph{types} (generated by $\btypes$) are defined by mutual recursion as follows
\[
\begin{array}{lrclr}
\hbox{\em Argument type} & \tau & ::= & \iota\mid\rho\\
\hbox{\em Relational type} &  \rho&::= & \boolsort\mid\tau\to\rho\\
\hbox{\em 1st-order type} & \sigma_{\FO}&::=&\iota\mid \boolsort\mid\iota\to\sigma_{\FO}\\
\hbox{\em Type} &  \sigma&::=&\rho\mid\sigma_{\FO},
\end{array}
\]
where $\iota\in\btypes$.
We sometimes abbreviate function types $\tau_1\to\cdots\to\tau_n\to\sigma$ to $\overline\tau\to\sigma$.
Intuitively, $\boolsort$ (where $\boolsort \notin \btypes$) is the type of the truth values (or Booleans).
The types $\sigma_{\FO}$ are exactly those of the form $\overline\iota\to\iota$ or $\overline\iota\to o$,
i.e.\ each argument is of some type $\iota_i \in \btypes$.
Moreover, each relational type has the form $\overline\tau\to o$.
We define $\order(\tau_1\to\cdots\to\tau_n\to\sigma)=n$ if $\sigma$ is $\iota$ or $\boolsort$.

\newcommand\sus[2]{{#1}^{(#2)}}
\newcommand{\sigs}{\overline{\Sigma}}

A \emph{type environment} (typically $\Delta$) is a function mapping
\emph{variables} (typically $x,y,z$) to {argument} types;
for $x \in \dom(\Delta)$, we write $x : \tau\in\Delta$ to mean $\Delta(x) = \tau$.
A \emph{signature} (typically $\Sigma, \Xi$) is a set of distinct typed \emph{symbols}
$c\from\sigma$, where $c\nin\dom(\Delta)$.
A signature $\Sigma$ is \emph{1st-order} if $\sigma$ is 1st-order for all $c\from\sigma\in\Sigma$.
We often write $c\in\Sigma$ if $c\from\sigma\in\Sigma$ for some $\sigma$.

The set of \emph{$\Sigma$-pre-terms} is given by
  $M::=x\mid c\mid M \, M$
where $c\in \Sigma$.
We assume that application associates to the left, 
and write $M\,\overline N$ for $M\,N_1\cdots N_n$, assuming implicitly that $M$ is not an application.


The typing judgement $\Delta\vdash M\from\sigma$ is defined by
\[
\begin{array}{c}
    \infer{\Delta\vdash x\from\Delta(x)}{x\in\dom(\Delta)}
    \quad
    \infer{\Delta\vdash c\from \sigma}{c\from\sigma\in\Sigma}
    \quad
    \infer{\Delta\vdash
      M_1M_2\from\sigma_2}{\Delta\vdash
      M_1\from\sigma_1\to\sigma_2&\Delta\vdash M_2\from\sigma_1}
  \end{array}
\]
\dw{Doesn't need to be in the main text (completely standard).} \lo{I agree. If necessary, this can go to the appendix.}

We say that $M$ is a \emph{$\Sigma$-term of type $\sigma$} if $\Delta\vdash M\from\sigma$. A $\Sigma$-term is a \emph{1st-order $\Sigma$-term} if the symbols in its construction are restricted to symbols $c\from\sigma_{\FO}\in\Sigma$
and variables ${x\from\iota}\in\Delta$.

\begin{remark}
  \label{rem:isimple}
   It follows from the definitions that
   each term $\Delta\vdash M\from \overline{\iota}\to\iota$
   can only contain variables of type $\iota_i$ and constants of non-relational 1st-order type (and contains no logical symbols, a similar approach is adopted in \cite{CHRW13}).
\end{remark}

We define a \emph{$\Sigma$-formula} $F$ by
\[
F ::=M\mid F\lor F\mid F\land F\mid \lnot F
\]
where $M$ is any $\Sigma$-term of type $\boolsort$. 
\dw{Is this used?}\lo{Yes; e.g.~when defining atom.}
For a $\Sigma$-term or $\Sigma$-formula $M$ and $\Sigma$-terms $N_1,\ldots,N_n$ and variables
$x_1,\ldots,x_n$ that satisfy $\Delta\vdash N_i\from\Delta(x_i)$, the 
\emph{substitution} $M[N_1/x_1,\ldots,N_n/x_n]$ is defined in the standard way.

\subsubsection{Semantics}
There are two classic semantics for higher-order logic:
\emph{standard} and \emph{Henkin semantics} \cite{DBLP:journals/jsyml/Henkin50}.
In this paper, we will not be concerned with the latter, but the notion of frame is useful.
Assume, for each $\iota\in\btypes$, an associated set $D_\iota$.
Formally, a \defn{frame} $\Hf$ assigns to each type $\sigma$ a nonempty set $\sinti\Hf\sigma$ such that
\begin{enumerate}
\item $\sinti\Hf \iota\defeq D_\iota$ for each $\iota\in\btypes$
\item $\sinti\Hf o\defeq\bool \defeq \{0,1\}$
\item For each $\sigma_1\to\sigma_2$,
$\sinti\Hf{\sigma_1\to\sigma_2}\subseteq[\sinti\Hf{\sigma_1}\to\sinti\Hf{\sigma_2}]$
\end{enumerate}
where $[U \to V]$ is the set of functions from (sets) $U$ to $V$.

\begin{remark}
\label{rem:weak-ca}
Unlike \cite{DBLP:conf/lics/OngW19}, we do not distinguish pre-frame and frame.
Because $\lambda$-abstractions are not part of the HoCHC syntax here, the (weak) comprehension axiom in \cite[p.~3]{DBLP:conf/lics/OngW19} does not apply. 
\end{remark}

%

\begin{example}[Standard frame] 
  We define the \emph{standard frame} $\Sf$ recursively by
  $\sinti{\Sf}
  o\defeq\bool$;
  $\sinti{\Sf}\iota\defeq D_\iota$
  for $\iota\in\btypes$; and
  \begin{align*}
    \sinti {\Sf}{\tau\to\sigma}&\defeq [\sinti {\Sf}\tau\to\sinti {\Sf}\sigma]
  \end{align*}
\end{example}

Let $\Sigma$ be a signature, and $\Hf$ be a frame.
A $(\Sigma,\Hf)$-\defn{structure} $\As$ assigns to each $c\from\sigma\in\Sigma$ an element
$c^\As\in\sinti\Hf\sigma$ and for convenience we set $\sinti\As\sigma\defeq\sinti\Hf\sigma$ 
for types $\sigma$.
A $(\Delta,\Hf)$-\defn{valuation} $\alpha$ is a
function such that for every $x\from\tau\in\Delta$,
$\alpha(x)\in\sinti\Hf\tau$.
For a $(\Delta,\Hf)$-valuation $\alpha$, variable $x$ and $r\in\sinti\Hf{\Delta(x)}$,
$\alpha[x\mapsto r]$ is defined in the usual way.

Let $\As$ be a $(\Sigma,\Hf)$-structure and let $\alpha$ be a $(\Delta, \Hf)$-valuation.
The \emph{denotation} $\sinti{\As} M(\alpha)$ of a $\Sigma$-term $M$ with respect to $\As$ and $\alpha$ is defined recursively by
\begin{align*}
\sinti{\As} x(\alpha)&\defeq \alpha(x)
\quad
\sinti{\As} c(\alpha) \defeq \assigned{c}{\As}
\\
\sinti{\As}{M_1 \, M_2}(\alpha)&\defeq \sinti{\As}{M_1}(\alpha)\big(\sinti{\As}{M_2}(\alpha)\big)
\end{align*}
For each term $\Delta\vdash M\from\sigma$, we have $\sinti{\As} M(\alpha)\in\sinti{\As}\sigma$.
(We will write $\sinti{\As} M(\alpha)$ as $\As^{\Sigma}\sem{M}(\alpha)$ when we need to be explicit about the signature of the $\Sigma$-terms $M$.)

\begin{example}[LIA]
	\label{ex:folstr}
	In this paper, many examples will use the signature of \defn{linear integer arithmetic}\footnote{with the usual types $0,1\from\iota$; ${+}, {-}\from\iota\to\iota\to\iota$ and $\triangleleft\from\iota\to\iota\to o$ for $\triangleleft\;\in\set{<,\leq,=,\neq,\geq,>}$; and we use the common abbreviation $n$ for $\underbrace{1+\cdots+1}_n$, where $1\leq n\in\nat$} (LIA) {(aka Presburger arithmetic)} $\Sigma_{\LIA} \defeq \{0,1,+,-,<,\leq,=,\neq,\geq,>\}$ and its standard model $\As_{\LIA}$.
\end{example}
\lo{NOTATION SUMMARY (useful for readers): Assume signature $\Sigma$, and given
\begin{itemize}
\item frame $\Hf$ (standard frame $\Sf$), which assigns meanings to types
\item $(\Sigma,\Hf)$-structure $\As$ (also $\Bs$), which assigns meanings to constant symbols (and types)
\item $(\Delta,\Hf)$-valuation $\alpha$ (where $\Delta$ is type environment), which assigns meanings to variables
\end{itemize}
the denotation of terms $\sinti{\As}{-}$ (and formulas $\sinti{\Bs}{-}$) is then determined: for each term $\Delta\vdash M\from\sigma$, we have $\sinti{\As} M(\alpha)\in\sinti{\As}\sigma$.}

\subsection{Higher-order constrained Horn clauses (HoCHC)}
\label{ch:defHoCHC}


%
%

%

We explicitly distinguish symbols of the \emph{background} (bg) theory from those---in the \emph{foreground} (fg)---which are constrained by clauses.
This distinction enables a certain semantic separation required by a model construction (\cref{def:oconfusionf}), which is crucial to our decidability result (\cref{thm:order-main-sat-result}). \footnote{Using notations in \cref{def:oconfusionf} and \cref{lem:order-finite}, take ${\le} \in \sigb$. If ${\assigned{\le}{\As}} \in \Hf\sem{\Z\to\Z\to\boolsort}$ then $\Hf\sem{\Z\to\boolsort}$ must be infinite, contradicting \cref{lem:order-finite}.}



\begin{assumption}\label{ass:sigma}
  Henceforth we fix a 1st-order signature $\sigb$,
  and a $(\sigb,\Sf)$-structure $\As$, and a finite signature $\sigf$ \emph{disjoint} from $\sigb$ with only predicate symbols (of a relational type), typically $X, Y, P$ and $R$ and their variants.
  We will write such a pair of signatures as $\sigs = (\sigb, \sigf)$.
\end{assumption}
Intuitively, $\sigb$ and $\As$ correspond to the language and interpretation of the background theory, e.g.\ $\Sigma_{\LIA}$ together with its standard model $\As_{\LIA}$.
In particular, \emph{we (only) consider background theories with a single model}.

Next, we introduce higher-order constrained Horn clauses and their
satisfiability problem \cite{popl18}.

\begin{definition}\label{def:atoms}
By \emph{atom}, we mean background atom or foreground atom.
  \begin{thmlist}
  \item A \emph{background atom} is a 1st-order $\sigb$-term of type $\boolsort$.
  \item A \emph{foreground atom} is a $\sigf$-term of type $\boolsort$.
  \end{thmlist}
\end{definition}

Note that a foreground atom has one of the following forms:
\begin{inparaenum}[(i)]
\item $R\,\overline M$ where $R\in \sigf$, or
\item $x\,\overline M$.
\end{inparaenum}

We use $\phi$ and $A$ (and variants thereof) to refer to background atoms and (general) atoms, respectively.
\tcb{Is $A$ used anywhere except the next line? If so, is it worth mentioning this here?}

\begin{definition}[HoCHC]\label{def:hochc}
  \begin{thmlist}
  \item A \defn{goal clause} (typically $G$) is a disjunction $\neg A_1\lor\cdots\lor\neg A_n$, where each $A_i$ is an atom.
  We write $\bot$ to mean the empty (goal) clause.

  \item If $G$ is a goal clause, $R\in \sigf$ and the variables in $\overline x$ are distinct, then $G\lor R\,\overline x$ is a \defn{definite clause}.
  \item A \defn{higher-order constrained Horn clause (HoCHC)} is a goal or definite clause.
  \end{thmlist}
\end{definition}

Throughout the document, we will often write a clause $\neg A_1\lor\cdots\lor\neg A_n\lor R\,\overline x$ as $R\,\overline x\impliedby A_1\land \cdots\land A_n$. \dw{Why do we need the disjunctive notation at all?} \lo{I think \cref{def:hochc} is fine, and we need to define the $\impliedby$ notation as it is used in examples.}

Next we give an example of HoCHCs from \cref{sec:intro}, explicitly listing the types involved and illustrating the structures.

\tcb{We should copy the updated version from the introduction (or just get rid of this)} \lo{I think it's  good to give an example of HoCHC. There is no need to repeat the clauses though.}
\begin{example}[A system of HoCHCs]
  \label{ex:HoCHC}
  Let $\sigb=\Sigma_{\LIA}\cup\{=_S\}$ and
  $\sigf=\{\mathrm{Iter}\from S \to \Z \to (S \to \Z \to o) \to o,
  \mathrm{Inc?}\from S \to S \to \Z \to o,
  \mathrm{Tw}\from S \to o \}$ and let $\Delta$ be a type environment satisfying
  $\Delta(m)=\Delta(n)=\Delta(k)=\Z$ and
  $\Delta(x)=\Delta(y)=\Delta(y')=S$ and
  $\Delta(p)=S\to\Z\to o$. The system consists of the HoCHCs in \cref{eg:tweet}, and the preceding three that define $\mathrm{Iter}$.
\end{example}

A \emph{$\sigs$-formula} is a formula where each term is either a $\sigf$-term or a 1st-order $\sigb$-term.
Let $\Hf$ be a frame that agrees with the standard frame $\Sf$ on the base types $\btypes$.
Let $\Bs$ be a $(\sigf,\Hf)$-structure and let $\alpha$ be a $(\Delta, \Hf)$-valuation.
The definition of the \emph{denotation} $\sinti{\Bs} F(\alpha)$ of a $\sigs$-formula $F$ with respect
to $\Bs$ and $\alpha$ is defined recursively 
by
  \begin{align*}
  \sinti{\Bs} M(\alpha) &\defeq
  \left\{
  \begin{array}{ll}
  \As^{\sigb}\sem{M}(\alpha_1) & \hbox{if $M$ a 1st-order {$\sigb$-term}}\\
  \Bs^{\sigf}\sem{M}(\alpha) & \hbox{if $M$ a {$\sigf$-term}}\\
  \end{array}
  \right.
  \\
  \sinti{\Bs}{F \land G}(\alpha)&\defeq \min(\sinti{\Bs}{F}(\alpha),\sinti{\Bs}{G}(\alpha))
  \\
  \sinti{\Bs}{F\lor G}(\alpha) &\defeq\max(\sinti{\Bs}{F}(\alpha),\sinti{\Bs}{G}(\alpha))
  \\
  \sinti{\Bs}{\neg F}(\alpha) &\defeq 1-\sinti{\Bs} F(\alpha)
  \end{align*}
where $\alpha_1$ is taken to be some $(\Delta,\Sf)$-valuation that agrees with $\alpha$ on the elements of $\Delta$ of type $\iota$. The choice of such $\alpha_1$ does not matter because it is only used to interpret 1st-order $\sigb$-formulas, which contain no variables from $\Delta$ that do not have type $\iota$ for some $\iota\in\btypes$.

For $\sigs$-formulas $F$, we write $\Bs, \alpha\models F$ if $\sinti\Bs F(\alpha)=1$, and $\Bs\models F$ if $\Bs,\alpha'\models F$ for all $\alpha'$.
We extend $\models$ in the usual way to sets of formulas.

\tcb{since $\As$ is fixed by \cref{ass:sigma}, we should probably drop it here, particularly since we just write "satisfiable" everywhere outside this section} \lo{Agreed, and dropped.}
\begin{definition}
  \label{def:sat}
  Let $\Set$ be a set of HoCHCs, and suppose $\Hf$ is a frame \changed[tcb]{which agrees with $\Sf$ on $\btypes$}.
  \begin{thmlist}
  \item \changed[tcb]{$\Set$ is $(\As,\Hf)$-\defn{satisfiable}  if there
    exists a $(\sigf,\Hf)$-structure $\Bs$ such that $\Bs\models\Set$.}
  \item $\Set$ is $\As$-\defn{satisfiable} (also called \emph{$\As$-standard-satisfiable})
    if it is $(\As,\Sf)$-satisfiable.
  \end{thmlist}
\end{definition}
Whilst the notion of $(\As,\Hf)$-satisfiability may seem obscure, it is sometimes easier to construct $(\sigf,\Hf)$-structures (cf.~\cref{lem:order-emtonore}); and for certain $\Hf$, $\As$-satisfiability implies $(\As,\Hf)$-satisfiability (cf.~\cref{lem:order-smtoem}).

\dw{Perhaps it is worth stating that semantic invariance does not quite apply, but $\As$-unsatisfiability implies $(\As,\Hf)$-unsatisfiability.} \dw{Comment added.}

\begin{definition}\label{def:program}
  A \defn{program} is a  finite set of definite clauses. 
\end{definition}

\begin{remark}
\label{rem:equivalent}
\begin{asparaenum}[(i)]
\item Under the definition of satisfiability, a program can be seen as a conjunction of clauses, universally quantified over variables in $\Delta$.
\item It is often convenient to write logically equivalent formulas such as $X\, x\impliedby \exists y. Y\, x\, y \lor Z\, x$ instead of $\{X\, x \lor \lnot Y\, x\, y,\ X\, z \lor \lnot Z\, z\}$.
We may even write the bodies of such formulas as existentially quantified formulas, over variables that do not appear in the head.
\item If $\sigb$ contains a predicate interpreted by $\As$ as equality (as with $\Sigma_{\LIA}$ and $\As_{\LIA}$), then we may write terms of integer type inside foreground atoms. For example $X\,(x+y+5)$ is equivalent to $X\,z\land (z=x+y+5)$ 
\item Every satisfiable set of clauses $\Gamma$ has, in each frame, a \defn{canonical model} $\Bs$ which arises by saturating under all immediate consequences \cite[Thm.~23]{DBLP:conf/lics/OngW19}.
In the higher-order setting, this model may not be least wrt inclusion, but for any goal clause $G$ we have: $\Gamma\cup \{G\}$ satisfiable iff $\Bs\models G$.
\end{asparaenum}
\end{remark}

\subsection{Resolution proof system}
\label{sec:resolution}
\dw{Didn't we discuss to reduce the discussion (in particular not state the proof system)?}
\tcb{I think we agreed to get rid of the example, but it's useful to have the rules here since they're used in \cref{lem:order-emtonore}}
We use a simple resolution proof system \cite{DBLP:conf/lics/OngW19} consisting of only two rules:
  \begin{enumerate*}
  \item a higher-order version of the usual resolution rule \cite{R65}
    between a goal clause and a definite clause (thus yielding a goal clause) and
  \item a rule to refute certain goal clauses which are not satisfied
    by the model of the background theory (similar to \cite{BGW94}).
  \end{enumerate*}

\medskip
\inferbintocl{Resolution}{\neg R\,\overline M\lor G}{G'\lor R\,\overline x}{G\lor \big(G'[\overline M/\overline x]\big)}


\inferuntocl{Refutation}{\neg x_1\,\overline M_1\lor\cdots\lor \neg x_m\,\overline M_m\lor\neg\phi_1\lor\cdots\lor\neg\phi_n}{\bot}

\medskip

\noindent 
With the latter rule applicable only there exists a valuation $\alpha$ such that $\As,\alpha\models\phi_1\land\cdots\land\phi_n$.
Since variables are implicitly universally quantified, we may assume $x_1\cdots x_m$ are interpreted as $(\bar{y}\mapsto \textsf{false})$.
The rules must be applied modulo  renaming of (free) variables;
we write $\Set'\Res\Set'\cup\{G\}$ if $G$ can be thus derived from the clauses in $\Set'$ using the above rules and $\Res^*$ for the reflexive, transitive closure of $\Res$.
\begin{restatable}[Soundness and Completeness \cite{DBLP:conf/lics/OngW19}]{theorem}{completeness}
  \label{thm:completeness}
  Let $\Set$ be a set of HoCHCs.
  Then $\Set$ is $\As$-unsatisfiable if, and only if, $\Set\Res^*\{\bot\}\cup\Set'$ for some $\Set'$.
\end{restatable}

It follows that a set of HoCHCs is $\As$-satisfiable if, and only if, it cannot be refuted by the proof system.

Consequently, the resolution proof system gives rise to a semi-decision procedure for the (standard) $\As$-unsatisfiability problem
provided the consistency\footnote{i.e.\ whether there exists a valuation $\alpha$ such that $\As,\alpha\models\phi_1\land\cdots\land\phi_n$} of conjunctions of atoms in the background theory is semi-decidable.

\section{Higher-order limit Datalog}
\label{sec:ldz}
In this section, we describe the limit restriction on HoCHC programs. We discuss the first-order fragment with this restriction, showing in \cref{sec:decidable-fold} that its satisfiability problem  is decidable. Then, in \cref{sec:undec-hold}, we show that this does not hold for higher-order problems, motivating the restrictions described in the rest of this paper.

To begin, we need some properties of the background theory, so we extend \cref{ass:sigma} by
\begin{assumption}\label{ass:ordering}
Henceforth fix some set $\WW$, a 1st-order signature $\Sigma_\WW$ and a $(\Sigma_\WW, \Sf)$-structure $\As_\WW$ such that the first-order theory of $(\Sigma_\WW,\As_\WW)$ is decidable, ${\le}\in{\Sigma_\WW}$, $\assigned{\le}{\As_\WW}$ is a preorder on $\WW$, and for each upset $X$ (i.e.~a subset of $\WW$ such that if $x\in X$ and $x\assigned{\le}{\As_\WW}y$ then $y \in X$), there is a $\Sigma_\WW$-formula $\phi(x)$ which expresses membership of $X$.

 \tcb{If $\assigned{\le}{\As_\WW}$ is not a preorder (and is not used anywhere except limit clauses), the satisfiability of the resulting program under $\As_\WW$ will be unchanged by changing $\assigned{\le}{\As_\WW}$ to its reflexive transitive closure (making it a preorder).}

Moreover fix some finite set $S$.
We strengthen \cref{ass:sigma} by asserting that:
  \begin{itemize}
  \item
  $\sigb := \Sigma_{\WW} \cup \{ {\eqS} : S \to S \to o \} \cup \{ s : S \mid s \in S\}$
  \item $\assigned{c}{\As}:= \assigned{c}{\As_{\WW}}$ if ${c : \sigma} \in \Sigma_{\WW}$; $\assigned{s}{\As} := s$ if $s \in S$; and $\assigned{(\eqS)}{\As}$ is the standard equality between elements of $S$. \qed
    \end{itemize}
\end{assumption}
Note that the constraints on $\WW$ imply that there are countably many upsets.
Examples of such structures include the integers with $\As_{\LIA}$ (the upsets are either $\Z$, $\emptyset$ or $\{x:x\ge k\}$ for some $k\in\Z$, each of which can easily be described by a formula) and any countable well-quasi-ordering (WQO) with a decidable background theory (for example, tuples of naturals under component-wise ordering also with the theory of linear integer arithmetic).
Also note that any predicate on $S$ can be expressed in terms of $\eqS$, so our examples may freely make use of other predicates.

Recall that a \emph{well-quasi-ordering} (WQO) \cite{SchmitzS2017} is a quasi-order $(W, \leq)$ such that every infinite sequence $w_1, w_2, \cdots$ contains an increasing pair: $w_i \leq w_j$ for some $i < j$.
To see that all upsets of a countable WQO $(\WW, \leq)$ are expressible, note that any upset $X\subseteq \WW$ has a finite number of minimal elements $m_1,\cdots,m_n$ (say), hence $X$ can be described as $\{w\in\WW : \bigvee_{i=1}^n w\ge m_i\}$. \changed[tcb]{Unfortunately, this does mean that all elements of $\WW$ must be constants in background theory, which makes it harder to obtain decidability (for example, the subword order is a WQO, but with constants, even the $\exists$-theory becomes undecidable \cite{DBLP:conf/lics/HalfonSZ17,DBLP:conf/fossacs/KuskeZ19})}.



\begin{remark} \label{rem:directiondoesntmatter}
	Note that the converse of a preorder is another preorder, and upsets under one are downsets (the complements of upsets) under the other.
  This means that if \cref{ass:ordering} holds for a relation, it also holds for its converse.
  We will make use of this by using upwards closed predicates in the definition below and in the proofs for technical convenience, despite the examples in \cref{sec:intro} using downwards closed predicates.
\end{remark}
\begin{definition} \label{def:hold} 
\begin{asparaenum}[(i)]
\item An \defn{(upwards) limit Datalog problem} $\Gamma$ is a finite set of HoCHC clauses over a signature $\sigs = (\sigb, \sigf)$,
\changed[tcb]{compatible with \cref{ass:sigma,ass:ordering}}, such that for each $X:\rho\in \sigf$, $\rho$ contains at most one argument of type $\WW$,
and if $\rho = \tau_1\to \ldots \to \tau_n\to \WW \to \tau_{n+1}\to\ldots \tau_{n+m}\to \boolsort$ then $\Gamma$ contains a \emph{limit clause}:
\[
X \, \overline{z} \, x \, \overline{z'} \; \impliedby \;
y \le x \land X \, \overline{z} \, y \, \overline{z'}
\]
writing $\overline{z}$ and $\overline{z'}$ for $z_1 \, \ldots \, z_n$ and $z_{n+1} \, \ldots \, z_{n+m}$ respectively.
\item A \defn{first-order limit Datalog problem} $\Gamma$ is a limit Datalog problem where for each $X:\rho\in \sigf$, it is the case that $\rho\in\sigma_{\FO}$,
and no atom that occurs in $\Gamma$ is headed by a variable.
\item The \defn{satisfiability problem} for limit Datalog asks: given a limit Datalog problem $\Gamma$, is it $\As$-satisfiable?.
\end{asparaenum}
\end{definition}

A key idea of limit Datalog is that predicates with $\WW$-typed arguments must be interpreted as sets that are closed upward with respect to that argument.  Consequently, a proposition $X\,\overline{z}\,y\,\overline{z'}$ asserts only that $X$ holds of ``at least $y$'' (i.e., $X$ is a min-predicate in the sense of \cite{DBLP:conf/ijcai/KaminskiGKMH17}).


\paragraph*{First-order limit \dz}
\cite{DBLP:conf/ijcai/KaminskiGKMH17} describe \firstdef{first-order limit \dz} which is first-order limit Datalog over linear integer arithmetic. Examples of this (\cref{eg:fo-tweet,eg:fo-paths}) are given in \cref{sec:intro}.

\begin{remark}
\cite{DBLP:conf/ijcai/KaminskiGKMH17} also allow predicates defining finite sets of integers;
and both min- and max-predicates;
and multiplication by constants, and by integers from fixed finite sets. They show that a limit \dz\ problem with these features can be transformed into one without them.
\end{remark}

%
%
%
%




First-order limit \dz{} is motivated by aggregation in declarative data analysis, which is typified by its requirements for recursion and linear integer arithmetic.
In declarative data analysis, the emphasis is on giving a specification of the required output rather than instructions on how to achieve it.
Such an analysis is enabled by a declarative language, and experience suggests that support for high-level programming over collection types (e.g. list comprehensions, map, reduce) is particularly beneficial \cite{ACCEHS10}. \tcb{We certainly can't support lists of integers (we can't even support pairs).}
Consequently, a higher-order foundation, such as HoCHC, may be particularly appropriate.

Of course higher-order programming is most important for larger codebases where it can be reused many times,
but \cref{eg:paths,eg:tweet} show that already the (first-order) examples given in \cite{DBLP:conf/ijcai/KaminskiGKMH17} have a shared structure that can be factored out using a higher-order combinator.





\subsection{Decidability at order 1}\label{sec:decidable-fold}

A key result of \cite{DBLP:conf/ijcai/KaminskiGKMH17} is that the decision problem for the first-order language is decidable.
We give an alternative proof of this theorem extended to first-order limit Datalog (allowing for structures other than linear integer arithmetic) which is helpful when understanding similar proofs in the sections that follow.

\begin{theorem}\label{thm:fold}
The satisfiability problem for first-order (upwards) limit Datalog is decidable.
\end{theorem}
\begin{proof}
Take a first-order limit Datalog problem $\Gamma$.
It follows from  \cref{def:hold} that for any $(\sigf,\Sf)$-structure $\Bs$ such that $\Bs\models \Gamma$, predicate $X:S^n \to \WW \to S^m\to \boolsort \in \sigf$, and tuples of constants  $\overline{s}, \overline{s'}\subseteq S$, the set \[
U = \{w \in \WW \mid \sinti\Bs{X \, \overline{s} \, x \, \overline{s'}}([x \mapsto w])=1\}
\]
is upwards closed. By \cref{ass:ordering}, there exists a 1st-order formula $\phi_{X,\overline{s}, \overline{s'}}(x)$ such that
\[
U = \{w \in \WW \mid \sinti{\As_\WW}{\phi_{X,\overline{s}, \overline{s'}}(x)}([x\mapsto w])\}.
\]

As there are finitely many predicates, finitely many tuples of elements of $S$ and countably many such formulas $\phi$, we have an r.e.~set of candidate models.
Given a \changed[sr]{$(\sigf,\Sf)$-structure} $\Bs$ of this form, we may ground all instances of variables from $S$, then substitute formulas $\phi_{X,\overline{s}, \overline{s'}}$ as appropriate, removing all instances of predicate symbols. Since the 1st-order theory of $(\WW,\As_\WW)$ is decidable, we can decide if $\Bs\models\Gamma$.


If $\Gamma$ is satisfiable, we can find such a structure by enumeration.
If not, then resolution (\cref{thm:completeness}) can prove that.
\end{proof}

\begin{corollary}
	The satisfiability problem for first-order downwards limit Datalog is decidable.
\end{corollary}
\begin{proof}
Although the above proof covers upwards limit Datalog, it only relies on the fact that upwards closed sets are expressible as 1st-order formulas.
Since the complement of every downwards closed set $D$ is an upwards closed set $U$, $D$ is described by the negation of the formula describing $U$.
Thus the proof also holds for downwards limit Datalog.
\end{proof}

\subsection{Undecidability in general}\label{sec:undec-hold}

Unlike the first-order case, higher-order limit Datalog in general is undecidable\footnote{The proof given here covers integers with linear integer arithmetic. A variant works for naturals, but higher-order limit Datalog is not undecidable for all structures $(\WW,\As_\WW)$.}, which can be proved by demonstrating that multiplication, hence Diophantine equations, is definable. 

The idea is to use a  pair of terms of type $\Z \to \boolsort$ to represent an integer.
Fix a higher-order limit \dz{} program $\Gamma$ and let $\Bs$ be its canonical model.
For an integer $k$, we write $\sem{(M,N)} \equiv k$ just if $\Bs\sem{M} = \{x : x \ge k\}$ and $\Bs\sem{N} = \{x : x \ge -k\}$. 
This ensures that, for any $n \in \Z$, $\Bs\sem{M\, n \wedge N\, (-n)} = 1$ iff $n = k$.
Then we say that a partial function $f : \Z^m \to \Z$ is \firstdef{definable in $\Gamma$} just if there exist two closed terms $M_1$ and $M_2$ of type
\[
  \underbrace{(\Z \to \boolsort) \to \cdots \to (\Z \to \boolsort)}_{2m\text{-times}} \to (\Z \to \boolsort)
\]
such that: if for each $i \in \{1, \ldots, m\}$, $\sem{(P_i,P_i')} \equiv k_i$ then
\[\begin{array}{ll}
& \sem{(M_1\, P_1\, P_1' \cdots{} P_m\, P_m', M_2\, P_1\, P_1' \cdots{} P_m\, P_m')} \\
\equiv & f(k_1,\,\ldots,\,k_m).
\end{array}
\]
\begin{example}[Addition]
	\label{eg:addition}
Consider the following program which defines addition and the constant 5. 
\begin{align*}
  \mathrm{Add}_1, \mathrm{Add}_2 &: \sigma \to \sigma \to \sigma \to \sigma \to \Z \to o\\
  \mathrm{I}_{51}, \mathrm{I}_{52} &: \sigma \quad \hbox{where }\sigma = \Z \to o
\end{align*}
\[
\begin{array}{l}
  \mathrm{Add}_1\, f_1\,f_2\,g_1\,g_2\,x \hspace{-10em}
  \\ \qquad\ \ \impliedby
  f_1\,y\land f_2\,(-y) \land g_1\,z \land g_2\,(-z)\land  x\ge y+z\\
  \mathrm{Add}_2\, f_1\,f_2\,g_1\,g_2\,x
  \\ \qquad\ \ \impliedby
  f_1\,y\land f_2\,(-y) \land g_1\,z \land g_2\,(-z)\land x\ge -(y+z)\\
  \mathrm{I}_{51}\, x
  \impliedby
  x\ge 5 \\
  \mathrm{I}_{52}\, x
  \impliedby
  x\ge -5
\end{array}
\]
\changed[sr]{
In the canonical model of this program,
\[
\mathrm{Add}_1\, \mathrm{I}_{51}\, \mathrm{I}_{52}\, \mathrm{I}_{51}\, \mathrm{I}_{52}\, x \wedge \mathrm{Add}_2\, \mathrm{I}_{51}\, \mathrm{I}_{52}\, \mathrm{I}_{51}\, \mathrm{I}_{52}\, (-x)
\]
would hold exactly when $x = 10$.
}
This means that the pair of partially applied functions $\mathrm{Add}_1\, \mathrm{I}_{51}\, \mathrm{I}_{52}\, \mathrm{I}_{51}\, \mathrm{I}_{52}$ and $\mathrm{Add}_2\, \mathrm{I}_{51}\, \mathrm{I}_{52}\, \mathrm{I}_{51}\, \mathrm{I}_{52}$ can be used as arguments to other functions; for example
\[
\mathrm{Add}_1\, (\mathrm{Add}_1\, \mathrm{I}_{51}\, \mathrm{I}_{52}\, \mathrm{I}_{51}\, \mathrm{I}_{52}) \, (\mathrm{Add}_2\, \mathrm{I}_{51}\, \mathrm{I}_{52}\, \mathrm{I}_{51}\, \mathrm{I}_{52}) \, \mathrm{I}_{51}\, \mathrm{I}_{52}\, x
\]
would hold for $x\ge 15$.
\end{example}

In App. \ref{apx:sec:ldz}, we give another example of how functions may be composed, and recursion can work, by defining multiplication.
With this we can define a goal clause corresponding to any Diophantine equation, in such a way that the program as a whole is satisfiable iff the equation has a solution.
Consequently:
\begin{theorem}[Undecidability]\label{thm:undecidability}
The satisfiability problem for higher-order limit \dz\ is undecidable.
\end{theorem}
\begin{proof}
Since solvability of Diophantine equations is undecidable \cite{MathematicalLogicForMathematicians}, so is the problem of determining if a higher-order limit \dz\ problem is satisfiable.
\end{proof}

\section{Initial limit Datalog}
\label{sec:order-initiality}

In this section, we prove \cref{thm:order-main-sat-result}, which says that a particular fragment of higher-order limit Datalog is decidable. The proof follows the same strategy as that of \cref{thm:fold}. The key difference occurs when we enumerate candidate models; even though we can restrict the first-order predicates to an enumerable set, there are still uncountably many inhabitants of higher-order types under standard semantics.

To work around this, first note that there are finitely many predicate symbols. If these were the only higher-order terms, we would be fine since
\cref{thm:completeness} can be seen as saying that satisfiability does not depend on the behaviour of predicates on elements of higher-order function spaces that don't correspond to terms. 
\tcb{I'd like to know if Dominik agrees with this phrasing}
\dw{Roughly: yes. Indeed the document (which was originally planned as a journal submission) makes this very precise. From what is written in the present paper I don't know whether this sentence makes sense to all readers. Furthermore, this only applies to non-background terms.}
However, terms can contain arbitrarily deeply nested subterms, as seen in \cref{eg:addition} (and in \cref{eg:multiplication} which is used in the proof of undecidability). This means there can be a countable infinity of terms with distinct interpretations, leading to an uncountable infinity of interpretations for predicates over those terms. We can prevent this kind of nesting by restricting the types of predicates in the following way.

We insist that among the arguments to a predicate {from $\sigf$}, at most one is of type $\WW$, and every argument that occurs to the left of the $\WW$-typed argument (if there is one) must be of a smaller order than this function of $\WW$.
For example, we would admit predicates of type $S \to \WW \to \boolsort$ and $(\WW \to o) \to \WW \to S \to (\WW \to \boolsort) \to \boolsort$, but not those of type $\WW \to \WW \to \boolsort$ nor $(\WW \to \boolsort) \to \WW \to \boolsort$.


\begin{definition}
\label{def:order-initial}
\begin{compactenum}[(i)]
\item An \defn{initial} type is a relational type $\sigma_1\to\cdots\to\sigma_n\to\boolsort$ where $n \geq 0$ satisfying
\begin{compactenum}[(O1)]
\item at most one of $\sigma_1,\cdots,\sigma_n$ is $\WW$, and
\item if $\sigma_j=\WW$ then for all $i<j$,
\(
\order(\sigma_i)<\order(\sigma_j \to \sigma_{j+1}\to\cdots\to\sigma_n\to\boolsort)
\), and
\item each $\sigma_j$ is $S$, or $\WW$, or initial.
\end{compactenum}
\item Let $\rho = \overline \sigma \to \boolsort$ be an initial type. We say that $\rho$ is an \defn{active} type (typically $\tyint$) if some $\sigma_i$ is $\WW$;
otherwise it is an \defn{inactive} type (typically $\tynint$).
\item An \defn{initial limit Datalog problem} is a limit Datalog problem where for each $X: \rho \in \sigf$, $\rho$ is initial.
\end{compactenum}
\end{definition}

\begin{example}
	All the types in \cref{eg:tweet,eg:paths,eg:tweet-follows,eg:summation} are initial; but neither  $\mathrm{Add}_1$ nor $\mathrm{Add}_2$ in \cref{eg:addition}
	have an initial type.
\end{example}

\tcb{
\begin{remark}
Condition (O3), which is included for technical convenience, is not needed in practice (although removing it does not make the expressive). If there is a predicate whose type violates (O3), there aren't any terms that will fit as an argument; hence it can't have any effect on decidability.
Thus we suppress (O3) in the discussion in \cref{sec:intro}.

There is some subtlety here - if there were predicates $X:(\WW\to\WW\to o)\to o$ and $Y:((\WW\to\WW\to o)\to o)\to o$ then $Y\; X$ would be a term that needs considering, so we appeal to the definition of satisfiability via resolution. Using a definition of initial without (O3) for the remainder of this paragraph, since in an initial limit Datalog problem there are no clauses where the head is a predicate of non-initial type, neither proof rule cares about such types any more than if they had type $o$.
 The easiest way to prove this formally would be to adjust the definition of an entwined frame (\cref{def:oconfusionf}) so that non-entwined structures are interpreted as the singleton set $\{\top\}$.
\end{remark}
}

\begin{theorem}[Decidability]
\label{thm:order-main-sat-result}
Given \cref{ass:sigma,ass:ordering}, there is an algorithm that decides whether a given initial limit Datalog problem is $\As$-satisfiable.
\end{theorem}

The initial type restriction does not prevent nested terms, but it does prevent problematic ones by making the subterm relationship compatible with the type-theoretic order of the terms involved.
If an order-$n$ term $N$ of an active type contains an order-$m$ subterm $M$ where $m> n$, then $M$ is a subterm of some $L$ (another subterm of $N$) of type $\sigma$ (say) whose order is less than $n$.
(This is because $N$ must have the form $X \, \overline{L}$ where $X \in \sigf$ is an active type, and each ${L_i}$ has order less than $n$.)
We will see later that we can take the interpretation of this type $\sigma$ to be a finite set, and hence we don't need to
know all possible values of $M$ to know all possible ways in which we can interpret $N$.


This ensures that we can enumerate candidate models up to their behaviour on constructible elements.  It allows (a) for interpretations to be defined inductively: the interpretation of all order-$n$ predicates is given before any of order-$(n+1)$ and (b) the behaviour of a predicate on a non-$\WW$ argument need only be specified on \emph{finitely many} definable elements (\cref{lem:order-finite}).

Consider a predicate symbol $X : (\WW \to o) \to o \in \sigf$.
Without restriction, there may be an infinity of definable elements of type $\WW \to o$ and hence uncountably many choices of interpretation of $X$.
However, a $\sigf$-term of type $\WW \to o$ can only be constructed by applying a predicate symbol $Y$ to some arguments $N_1, \cdots{}, N_k$.
It follows that $Y$ has a type of shape $\sigma_1 \to \cdots{} \to \sigma_k \to \WW \to o$.
By the initial type restriction, each $\sigma_i$ is necessarily $S$ and hence finite.
If we have already fixed the interpretation of each such $Y$ (each being of lower order than $X$), then there are only finitely many definable elements at type $\WW \to o$.
Hence, there are only finitely many definable relations at type $(\WW \to o) \to o$.

Of course, when first fixing the interpretation of $Y$ there can be infinitely many choices; but thanks to the limit Datalog restriction, only countably many can satisfy the limit clause which requires that any such interpretation is upward-closed in its $\WW$ argument.
It is straightforward to see that the choices are, moreover, r.e. (\cref{lem:order-re}).


This leads to the notion of an interpretation that is built up inductively by order, in which the domains of the higher-order predicates (i.e. interpretation of types) are not determined until the interpretations of lower-order predicates have been fixed.
The process of choosing interpretations for the types (i.e. the frame) and the process of choosing interpretations for the predicate symbols are entwined.

\paragraph*{Assumptions} Recall disjoint signatures $\sigb$ and $\sigf$ and 1st-order structure $\As$ from \cref{ass:ordering,ass:sigma}.
Henceforth fix an initial limit Datalog problem $\Gamma$ and take 
\[
l := \max\{ \order(\rho) \mid X : \rho \in \sigf\}.
\]

\begin{definition}
	\label{def:expansion-of-Bs}
	Let $\Xi_1$ and $\Xi_2$ be (possibly higher-order) signatures such that $\Xi_1 \subseteq \Xi_2$;
	and $\Hf_1$ and $\Hf_2$ be frames.

	Suppose $\Bs_1$ is a $(\Xi_1, \Hf_1)$-structure.
	We say that a $(\Xi_2, \Hf_2)$-structure $\Bs_2$ is a $(\Xi_2, \Hf_2)$-\defn{expansion} of $\Bs_1$ just if $c^{\Bs_2} = c^{\Bs_1}$ for all $c \in \Xi_1$.
\end{definition}

We first define, given sets $U_1, \cdots, U_n$, a relation on relational functions, ${\le_{\boolsort, n}} \subseteq [U_n \to \cdots \to U_1 \to \mathbb{B}]^2$, by
\begin{align*}
f \le_{\boolsort, 0} g &\defeq (f = 0 \textrm{ or } g = 1)\\
f \le_{\boolsort, n+1} g &\defeq \forall x \in U_{n+1} \, . \, f (x) \le_{\boolsort, n} g (x)
\end{align*}
Define $\top_n \in [U_n \to \cdots \to U_1 \to \mathbb{B}]$ as $\forall \overline x \, . \, \top_n \, \overline x = 1$.
Henceforth we elide the subscript $n$ from $\le_{\boolsort, n}$ and $\top_n$.

\begin{definition}
\label{def:oconfusionf}
Let $n \geq 1$ and $\Xi \subseteq \sigf$.
Given a $(\Xi,\,\Hf)$-structure $\Bs$, define the \defn{\oconfused order-$n$ frame derived from $\Bs$}, written $\oconfusionf{\Bs}_n$, by case analysis of $\sigma$ as follows.
\begin{compactenum}[(i)]
\item $\sigma$ is initial and $\order(\sigma) \le n-2$, or $\sigma$ is (inactive, or $S$, or $\WW$, or $\boolsort$) and $\order(\sigma)\le n-1$:
\[
\oconfusionf{\Bs}_n\sem{\sigma} := \Hf\sem{\sigma}
\]
\item $\sigma$ active and $\order(\sigma)=n-1$: 
\begin{align*}
\oconfusionf{\Bs}_n\sem{\WW\to\tynint} &:=
    \big\{\top \in \fspace{\WW}{\oconfusionf{\Bs}_n\sem{\tynint}}\big\}\ \cup\\
    &\hspace{-20pt}\big\{\assigned{X}{\Bs} \, \overline{s} \, \mid \, X: \overline{\tau} \to \WW \to\tynint \in {\Xi},
    \ s_i \in \Hf\sem{\tau_i}\big\}  \\
\oconfusionf{\Bs}_n\sem{\sigma_1\to\tyint} &:= \fspace{\oconfusionf{\Bs}_n\sem{\sigma_1}}{\oconfusionf{\Bs}_n\sem{\tyint}} 
\end{align*}

\item $\sigma$ is initial and $\order(\sigma)= n$:
\begin{align*}
\hspace{-20pt}\oconfusionf{\Bs}_n\sem{\WW\to\tynint}
&:=\\
&\hspace{-30pt}\big\{f\in\fspace{\WW}{\oconfusionf{\Bs}_n\sem{\tynint}} \mid \forall z\assigned{\le}{\As} z' .f(z)\le_\boolsort f(z') \big\}\\
\hspace{-20pt}\oconfusionf{\Bs}_n\sem{\sigma_1\to\sigma_2} &:= \fspace{\oconfusionf{\Bs}_n\sem{\sigma_1}}{\oconfusionf{\Bs}_n\sem{\sigma_2}} \quad (\sigma_1 \not= \WW)
\end{align*}

\item $\sigma$ is not initial, or $\order(\sigma) > n$:
\[
\oconfusionf{\Bs}_n\sem{\sigma_1\to\sigma_2}  := \fspace{\oconfusionf{\Bs}_n\sem{\sigma_1}}{\oconfusionf{\Bs}_n\sem{\sigma_2}}
\]
\end{compactenum}
\end{definition}



The preceding definition is used in the context of entwined structures (\cref{def:cconfused-family}).

In that context, we explain the cases:
By \cref{lem:order-finite}, sorts covered by case (i) can be treated as finite, so are easy to deal with.
Case (ii) is the most interesting - it is where we make use of the structure $\Bs$. Here we set the interpretation of order-$(n-1)$ active types $\WW\to\nu$ to be the minimum ensuring that $\Bs$ is still a $\oconfusionf{\Bs}_n$-structure. Top ($\top$) is needed for technical reasons. 
Case (iii) keeps things countable by discarding interpretations that don't satisfy the limit clauses. This is exactly like the proof of \cref{thm:fold}. This defines the space from which we will pick interpretations for order-$n$ predicate symbols (say $\Bs'$), and $\oconfusionf{\Bs'}_{n+1}$ will fix the space to become finite.
Case (iv) of \cref{def:oconfusionf}
is only there to ensure that $\oconfusionf{\Bs}_n$ is technically a frame. Such types are not used anywhere.

For $i \geq 1$, let \changed[lo]{$\Sigma_i\subseteq\sigf$} consist of the predicate symbols of \changed[lo]{$\sigf$} with types of order at most $i$.

\begin{definition}
\label{def:cconfused-family}
\begin{compactenum}[(i)]
\item A family of structures $\{\Bs_n\}_{n \in \omega}$, indexed by (order) $n$, is said to be \defn{\oconfused} just if $\Bs_0$ is the unique $(\emptyset,\Sf)$-structure, and
each $\Bs_{n+1}$ is a $(\Sigma_{n+1},\oconfusionf{\Bs_{n}}_{n+1})$-expansion of $\Bs_n$.

\item An \defn{\oconfused structure} is a member of some \oconfused family.
An \firstdef{\oconfused model} of $\Gamma$ is an \oconfused structure {$\Bs_{l+1}$} such that $\Bs_{l+1}\models\Gamma$.
\lo{N.B. $l$ is defined in \cref{ass:order-initial-limit-dz}.}
\end{compactenum}
\end{definition}


\begin{example}\label{eg:xyz} 
Using the background theory LIA (so $\WW=\mathbb{Z}$),
take, for example, the term $X\, a\, (Y\, b)\, (Z\, 5\, X)$ for some $a,b$ such that $\Delta(a),\Delta(b)= S$ and
  \begin{align*}
  Y &: S\to S\to\boolsort\\
  X &: \rho = S\to (S\to\boolsort) \to\boolsort\to\WW\to (\WW\to\boolsort) \to\boolsort\\
  Z &: \WW\to \rho \to\boolsort
  \end{align*}
  Now $Z$ has a complicated type, but $Z\, 5\, X$ must be either true or false, so we can select behaviours for $X$ ignorant of $Z$ (and the choices for $Z$ can depend on this without introducing a problematic cycle).
  \changed[sr]{This example is elaborated in \cref{sec:entwined-example} of the appendix.}
  \tcb{Should this be more closely tied to the earlier description mentioning predicates $X$ and $Y$ with slightly different types?}
\end{example}


In the following lemmas, let {$\{\Bs_n\}_{n \in \omega}$} be an \oconfused family, and set $\Hf_n := \oconfusionf{\Bs_{n-1}}_n$.
Note that $\Sigma_{l+1} = \sigf$.

\begin{restatable}{lemma}{orderfinite}
    \label{lem:order-finite}
Let $\sigma$ be an initial type. If $n>\order(\sigma)$, or $n=\order(\sigma)$ and $\sigma$ is an inactive type, then $\Hf_n\sem{\sigma}$ is finite.
\end{restatable}

\begin{restatable}{lemma}{orderonre}\label{lem:order-lvnre}
    Let $\sigma$ be an initial active type. If $n=\order(\sigma)$ then $\bnsem{\sigma}$ is r.e.
\end{restatable}

\begin{restatable}{lemma}{orderre}
    \label{lem:order-re}
    The set of \oconfused families of structures is r.e.
\end{restatable}

For each resolution proof rule, if $\Bs_{l+1}$ entails the premises of the rule, then it entails the conclusion. Since $\Bs_{l+1}\sem{\bot}=0$, there is no resolution proof of $\bot$.

\begin{restatable}{lemma}{orderemtonore}
\label{lem:order-emtonore}
If there is an \oconfused family such that $\Bs_{l+1}$ models $\Gamma$, then there is no resolution proof of $\bot$ from $\Gamma$.
\end{restatable}

\changed[sr]{
    On the other hand, the inductive construction gives enough freedom to choose appropriate interpretations for the predicate symbols whenever the clauses are satisfiable.  Any model can be reconstructed as an entwined structure that also satisfies the clauses, with the relationship between the two mediated by a logical relation.
}

\begin{restatable}{lemma}{ordersmtoem}
	\label{lem:order-smtoem}
If\; $\Gamma$ is satisfiable then there is an \oconfused structure that models $\Gamma$.
\end{restatable}


Observe that, for any \oconfused family, $\Hf_{l+1}\sem{\rho}$ is finite whenever $X:\rho\in\sigf$.
We can ask whether a particular $\Bs_{l+1}\models \Gamma$ and
this is decidable because it is equivalent to a formula in the first-order theory of $\Sigma_\WW$.

\begin{restatable}{lemma}{orderdecidablemodel}\label{lem:order-decidablemodel}
	Given an \oconfused structure $\Bs_{l+1}$, determining if it satisfies a goal or definite clause $G$ is decidable.
\end{restatable}

\paragraph*{Proof of \cref{thm:order-main-sat-result}}
If there is a refutation of $\Gamma$ by resolution, then we know $\Gamma$ is $\As$-unsatisfiable. By \cref{lem:order-emtonore}, there is no \oconfused structure $\Bs_{l+1}$ such that $\Bs_{l+1}\models\Gamma$.

If there is no resolution proof of $\bot$, then there is some model for $\Gamma$ in standard semantics. This model can be converted into an \oconfused model by \cref{lem:order-smtoem}. Hence enumerating \oconfused structures---possible because they are r.e. (\cref{lem:order-re}) and determining if $\Bs_{l+1}\models\Gamma$ is decidable (\cref{lem:order-decidablemodel})---will find a model.

Therefore, we may interleave a search for resolution proofs of $\bot$ with a search for \oconfused models resulting in a decision procedure for the initial limit Datalog decision problem.
\qed

\begin{example}
For a concrete example of an entwined structure, see Appendix \cref{sec:entwined-example}.
\end{example}

\newcommand{\FOLD}{first-order limit \dz{}}
\newcommand{\HOILD}{higher-order initial limit \dz{}}

\begin{remark}[Higher-order initial limit \dz{}]
It can be shown that higher-order initial limit \dz{} is strictly more expressive than first-order limit \dz{}. By this, we mean that there are queries about databases (aka structures on finite sets) that can be expressed with \HOILD{} but not \FOLD{}.
This follows from a result in \cite{DBLP:journals/tplp/CharalambidisNR19} which shows that the data complexity of $k$-order Datalog lies in $(k-1)$-EXPTIME. Since this only uses finite sets, the programs involved are valid \HOILD{} programs. \cite{DBLP:conf/ijcai/KaminskiGKMH17} shows that \FOLD{} has more reasonable time bounds (coNP-complete in database size) hence it must be less expressive.
\end{remark}

\section{Examples}
\label{sec:examples}

In this section, we give examples of how first-order limit Datalog problems can be used with the background theory of tuples of naturals (we use currying to avoid explicitly specifying projection functions), and other WQOs.

\subsection{Theory of tuples of naturals}

In the context of limit Datalog, the theory of tuples of naturals with componentwise ordering is much more powerful than that of integers as demonstrated by the examples below, the latter of which could not be accomplished using limit \dz.

The following set of clauses express multiplication:
\begin{align*}
	F &: S\to \mathbb{N}\to\mathbb{N}\to\boolsort\\
	F\,s\,x\,y &\impliedby  x\ge 0\land y\ge a \land D( s, a, b)\\
	F\,s\,x\,y &\impliedby  y+1\ge n\land F\,r\,n \land x\ge r + b \land D(s,a,b)\\
	G\, x &\impliedby  F\,s\,x\,0
\end{align*}
Here $D$ is a database predicate and we assume that for each $s\in S$ there is a unique pair of natural numbers $a,b$ such that $(s,a,b)\in D$.
In the canonical model of this set of clauses, the interpretation of $G$ is the set $\{x\in\mathbb{N}:\exists a,b,s. (s,a,b)\in D \land x\ge ab\}$.

This does not lead to undecidability like \cref{eg:multiplication} because this only expresses multiplication of constants, not variables.

This may be extended to express exponentiation $\alpha^\beta$ (demonstrated below) and further to other hyperoperations.
\lo{What are hyperoperations?}\tcb{\url{https://en.wikipedia.org/wiki/Hyperoperation} , a sequence that begins "addition, multiplication, exponentiation" related to the Ackermann function}
\begin{align*}
	F &:  \mathbb{N}\to\mathbb{N}\to\mathbb{N}\to\boolsort\\
	F\,x\,y\,z &\impliedby  x\ge 0\land y\ge \alpha\land z\ge 0\\
	F\,x\,y\,z &\impliedby  x\ge 1\land y\ge 0\land z\ge\beta\\
	F\,x\,y\,z &\impliedby  z+1\ge n\land F\,d\,0\,n\land y+1\ge m\land {}\\
	&\qquad F\,r\,m\,z \land x\ge r + d\\
	G\, x &\impliedby  F\,x\,0\,0
\end{align*}

\subsection{Lossy counter machines}

\newcommand{\inccmd}[3]{#1:\ #2 \mathrel{:=} #2 + 1;\ \mathsf{goto}\ #3}
\newcommand{\deccmd}[4]{#1:\ \mathsf{if}\ #2 = 0\ \mathsf{then\ goto}\ #3\ \mathsf{else}\ c_i \mathrel{:=} c_i - 1;\ \mathsf{goto}\ #4}

A \emph{(classic) lossy $n$-counter machine ($n$-LCM)}, due to \cite{MAYR2003337}, consists of: a finite set of states $Q$, an initial state $q_0 \in Q$, a final state $q_f \in Q$, $n$ counters $c_1,\ldots,c_n$, and a finite set of instructions, each of one of the two shapes A or B: 
\begin{description}
\item[A.] $(\inccmd{q}{c_i}{q'})$
\item[B.] $(\deccmd{q}{c_i}{q'}{q''})$
\end{description}

A \emph{configuration} $s$ of such a machine is an $(n+1)$-tuple of shape $(q,m_1,\ldots,m_n)$ where $q \in Q$ and each $m_i \in \mathbb{N}$ being the current value of counter $c_i$.  

\newcommand{\lossy}{\xrightarrow{l}}

A \emph{transition} of such a machine consists of spontaneous loss, followed by the execution of an instruction, followed by spontaneous loss:
\[
  s_1 \Rightarrow s_2 \quad\text{iff}\quad \exists s_1',\,s_2'.\ s_1 \lossy s_1' \to s_2' \lossy s_2
\]
The execution of an instruction $(p,m_1,\ldots,m_i,\ldots,m_n) \to (p',m_1,\ldots,m_i',\ldots,m_n)$ is defined iff:
\begin{itemize}
   \item there is an instruction of shape A and $p=q$, $m_i' = m_i + 1$ and $p'=q'$
   \item or, there is an instruction of shape B and $p=q$, $m_i = 0$, $m_i' = 0$ and $p'=q'$
   \item or, there is an instruction of shape B and $p=q$, $m_i > 0$, $m_i' = m_i - 1$ and $p'=q''$.
\end{itemize}
The spontaneous (classic) loss $(q,m_1,\ldots,m_n) \lossy (q,m_1',\ldots,m_n')$ is defined iff $\forall i \in [1,n].\ m_i' \leq m_i$.
Let us write $\Rightarrow^*$ for the reflexive, transitive closure of the transition relation.

The \emph{reachability problem} for $n$-LCM is to decide the following: given a configuration $s$, does $(q_0,0,\ldots,0) \Rightarrow^* s$?  It is known that the reachability problem is decidable as a special case of \cite{bouajjani-mayr-stacs99}.  Here we give an alternative approach using initial limit Datalog over tuples of natural numbers.  

To decide the problem, it suffices to construct a set of definite clauses $C$ over the foreground signature 
\[
   \{R_q : \mathbb{N} \to \cdots{} \to \mathbb{N} \to o \mid q \in Q\}
\] 
(each $R_q$ is of arity $n$) with canonical model $\mathcal{B}$, in such a way that, for each state $q$, we have $\mathcal{B} \models R_q\ m_1\ \cdots m_n$ iff $(q_0,0,\ldots,0) \Rightarrow^* (q,m_1,\ldots,m_n)$.
We define $C$ as follows, abbreviating $x_1 \cdots{} x_n$ and $y_1 \cdots{} y_n$ by $\vec{x}$ and $\vec{y}$ respectively.
\begin{compactitem}
   \item The clause $R_{q_0}\,\vec{x} \leftarrow \bigwedge_{i \in [1,n]} x_i = 0$ is in $C$.
   \item For each state $q \in Q$, the following limit clause is in $C$:
   \[
      R_{q}\ \vec{x} \leftarrow R_q\ \vec{y} \wedge \bigwedge_{i \in [1,n]} x_i \leq y_i
   \]
   \item For each instruction of shape A, the following clause:
     \[
     	R_{q'}\ \vec{x} \leftarrow
     	R_{q}\ \vec{y} \wedge x_i = y_i + 1 \wedge \bigwedge_{j \in [1,n] \setminus \{i\}} x_j = y_j
     \]
   \item For each instruction of shape B, the two clauses:
   \begin{align*}
         R_{q'}\ \vec{x} &\leftarrow
         R_{q}\ \vec{y} \wedge y_i = 0 \wedge \bigwedge_{j \in [1,n]} x_j = y_j \\
         R_{q''}\ \vec{x} &\leftarrow
         R_{q}\ \vec{y} \wedge y_i > 0 \wedge x_i = y_i - 1 \!\!\bigwedge_{j \in [1,n] \setminus \{i\}} \!\!x_j = y_j
   \end{align*}
\end{compactitem}

\paragraph*{Lossy channel systems and other WSTSs} Lossy counter machines are an example of a well structured transition system (WSTS) \cite{DBLP:journals/tcs/FinkelS01}. Other examples of these, such as lossy channel systems (LCSs), also have decidable reachability problems, but these do not immediately fall under our theorem because the relevant first-order theories are not decidable (in the case of LCSs  the relevant theory is that of strings with concatenation with constants and the subword ordering). In some cases these results can be proved by inspecting details of exactly where in our proof the decidability property is required. 

Part of our result is subsumed by the decidability of the coverability problem for WSTSs - specifically the first-order fragment where clauses only have a single foreground atom in the body and the background theory is a WQO.

\subsection{Languages ordered by the subword order}

\newcommand\sqsubsetdot{\mathbin{{\sqsubset}\mkern-6mu{\cdot}}}

The subword relation is a simple and important example of a WQO.
\cite{DBLP:conf/fossacs/KuskeZ19} study the decidability of first-order theories (and fragments thereof) of languages 
with the subword order.
Recall that a language $L \subseteq \Delta^\ast$ is \emph{bounded} if $L \subseteq w_1^\ast \, \cdots \, w_n^\ast$ for some $n \geq 0$, and $w_1, \cdots, w_n \in \Delta^\ast$. 
Consider structures of the form $(L, \sqsubseteq, (w)_{w \in L})$ for some $L \subseteq \Sigma^\ast$ where $\sqsubseteq$ is the subword relation, and we can use every word from $L$ as a constant.
\lo{The idea is to argue in terms of vectors $(x_1, \cdots, x_n) \in \mathbb{N}^n$ for which $w_1^{x_1} \cdots w_n^{x_n} \in L$.}

\begin{theorem}[Kuske and Zetzsche \cite{DBLP:conf/fossacs/KuskeZ19}]
\label{thm:kuske and zetzsche}
Let $L \subseteq \Delta^\ast$ be bounded and context free.
Then the first-order theory of $(L, \sqsubseteq, (w)_{w \in L})$ is decidable.
\end{theorem}

The theorem in fact holds for (a larger signature, and) a more expressive logic, first-order logic extended by a modulo counting quantifier \cite{DBLP:conf/fossacs/KuskeZ19}.
The proof is by interpreting the structure in Presburger arithmetic, $(\mathbb{N}, +)$, which is known to be decidable in this logic.

Since $(L, \sqsubseteq)$ is a countable WQO, it follows from \cref{thm:kuske and zetzsche} and \cref{thm:order-main-sat-result} that the associated initial limit Datalog problem is decidable.

\subsection{Basic process algebras and pushdown systems}

An important class of countable WQO  are \emph{context-free processes} (or \emph{basic process algebra}) and the more general collection of \emph{pushdown systems}, with respect to the subword ordering \cite{DBLP:journals/tcs/FinkelS01}; 
moreover they are \emph{automatic structures} (folklore but see e.g.~\cite{Lin2010,Brarany2007}) and so have decidable first-order theories (\cite{DBLP:journals/tcs/Hodgson82,DBLP:conf/lcc/KhoussainovN94} and various others).
It follows that they satisfy \cref{ass:ordering}.

\lo{Develop this!}
\lo{Possibly relevant: \cite{DBLP:conf/concur/ChadhaV07}.}

\section{Related work and further directions}
\label{sec:related}

\subparagraph{Decidable classes of constrained Horn clauses}
Cox, McAloon and Tretkoff \cite{DBLP:journals/amai/CoxMT92} have shown a catalogue of sub-recursive complexity results for various fragments of CHC obtained by restricting the syntax (in particular, the placement of variables) and the mechanism by which parameters are passed.  
Our work, however, takes Kaminski, Cuenca Grau, Kostylev and Motik's limit restriction \cite{DBLP:conf/ijcai/KaminskiGKMH17} as the starting point.

The limit restriction was introduced as a way of taming the undecidability of \dz{} \cite{DEGV01} that was compatible with the desire to express problems in declarative data analysis.
Moreover, it is shown in \cite{DBLP:conf/ijcai/KaminskiGKMH17}~that, under reasonable assumptions, the data complexity of the entailment in the logic is PTIME.
Our work extends limit \dz{} to higher-orders.
Higher-order extensions of Datalog are interesting in their own right: \cite{DBLP:journals/tplp/CharalambidisNR19} have shown that, on ordered databases, order-$k$ Datalog captures ($k-1$)-EXPTIME.

\subparagraph{Decidability beyond first order}
There is a lot of interest in the decidability of theories that go beyond first-order logic.
A very well studied case is that of monadic second-order theories (see e.g. \cite{Gurevich85}).
Of these, perhaps the best known is Rabin's celebrated result on the decidability of the theory of two successor functions \cite{rabin1969decidability}, from which the decidability of several other monadic second-order theories can be deduced.

For applications in e.g. higher-order program verification, however, it is important to retain higher-type relations of all arities and to admit background theories.
A recent work with similar requirements is that of \cite{MMSV18} who, motivated by applications in program synthesis, have introduced the logic \textsf{EQSMT}.
Formulas of this logic have a $\exists^*\forall^*$ prefix supporting second-order quantification at certain types.
They show that satisfiability of \textsf{EQSMT} formulas is decidable whenever satisfiability for the relevant fragments of the background theories is decidable.

\subparagraph{Higher-order constrained Horn clauses}
Our work takes place in the setting of HoCHC \cite{popl18}.  Even when the background theory is decidable, satisfiability of HoCHC is typically undecidable (already, first-order constrained Horn is typically undecidable \cite{DEGV01}).  However \cite[\S~VIII]{DBLP:conf/lics/OngW19} identified the so-called Bernays-Sch\"onfinkel-Ramsey fragment of HoCHC, modulo a restricted form of linear integer arithmetic, has a decidable satisfiability problem by showing {equi-satisfiability} to clauses w.r.t.\ a finite number of background theories with finite domains.
\changed[lo]{(HoCHC satisfiability is decidable for trivial background theories (e.g.~those of finite domains).)}

An alternative higher-order logic supporting integer arithmetic is $\textrm{HoFL}_Z$ of \cite{HFLZ}.
Whilst we do not know of any work on decidable fragments of $\textrm{HoFL}_Z$, \changed[sr]{we expect that} a version of our results on initial limit \dz{} could be transposed into that setting.

\lo{TODO. Discuss recent by Kobayashi et al.~involving circular proofs.}


\paragraph*{Future directions}

One question that remains open is: for which sets of types is the higher-order limit \dz\ problem decidable when predicates are restricted to those types?
There are alternatives, broadly similar to \cref{def:order-initial}, which are neither a superset nor a subset of the set of initial types, for which the same proof strategy works \changed[tcb]{(and we conjecture that such results can be proved as corollaries to \cref{thm:order-main-sat-result}, by inserting dummy variables)}.


Except for the lower bounds due to being a superset of higher-order Datalog, we have not considered runtime complexity of this problem. If the algorithm derived from the decidability proof were used, calculating its runtime would be an exercise in the construction of large numbers. Since many practical uses would have shapes that could be converted to 1st-order programs, there is some hope for tractable performance on useful subsets of initial limit \dz{}.






\paragraph*{Conclusion}

We have presented \emph{initial limit Datalog}, the first higher-order extension of constrained Horn clauses (over a non-trivial background theory) for which the satisfiability problem is decidable. 
Moreover the decision procedure extends to a variety of background theories, including linear integer arithmetic, and any countable well-quasi-order with a decidable first-order theory.
Our decidability proof uses a new kind of term model, called \emph{entwined structures}, which are recursively enumerable, and model checking is decidable.


\bibliographystyle{IEEEtran}
\bibliography{bibfile}

\newpage

\ifreview

\end{document}

\else

\appendices






\section{Supplementary materials for \cref{sec:prelims}}
\label{apx:sec:prelims}

\subsection{Logical relations}
\label{sec:logical-relations}
\begin{definition}\label{def:logicalrelations}
  Let $\Hf$ and $\Hf'$ be frames. A family of relations ${\arel_\sigma}\subseteq\Hf\tden\sigma\times\Hf'\tden\sigma$ are \emph{logical} if for all types $\tau\to\sigma$, $r\in\Hf\tden{\tau\to\sigma}$ and $r'\in\Hf'\tden{\tau\to\sigma}$, $r\arel_{\tau\to\sigma} r'$ iff for all $s\in\Hf\tden\tau,s'\in\Hf'\tden\tau$, 
  if $s\arel_\tau s'$ then $r(s)\arel_\sigma r'(s')$.
\end{definition}
Extending the $\arel_\tau$ in the usual pointwise fashion to frames and valuations we obtain:

\begin{restatable}{lemma}{termmonapp}
  \label{lem:termmonapp}
  Let $\Bs$ be a $(\sigf,\Hf)$-structure and $\Bs'$ be a $(\sigf,\Hf')$-structure, $\alpha$ be a $(\Delta,\Hf)$-valuation, $\alpha'$ be a $(\Delta,\Hf')$-valuation and $M$ be a term.

  If $\Bs\arel\Bs'$ and $\alpha\arel\alpha'$ then $\sinti\Bs M(\alpha)\arel\sinti{\Bs'} M(\alpha')$.
\end{restatable}
\begin{proof}
    We prove the claim by induction on the structure of $M$.
  
    The cases for variables and symbols from the foreground signature follow immediately from the assumptions.
  
    If $M$ is an application $M_1\,M_2$ then by the inductive
    hypothesis $\sinti\Bs{M_1}(\alpha)\arel\sinti{\Bs'}{M_1}(\alpha')$ and
    $\sinti\Bs{M_2}(\alpha)\arel\sinti{\Bs'}{M_2}(\alpha')$. Therefore, by definition of ${\arel}$,
    \begin{align*}
      \sinti\Bs{M}(\alpha)&=\sinti\Bs{M_1}(\alpha)(\sinti\Bs{M_2}(\alpha))\\
                          &\arel \sinti{\Bs'}{M_1}(\alpha')(\sinti{\Bs'}{M_2}(\alpha'))
                            =\sinti{\Bs'}{M}(\alpha')
    \end{align*}
  \end{proof}

\begin{restatable}{corollary}{lrposterm}
	\label{cor:lrposterm}
	Let $\Bs$,  $\Bs'$, $\alpha$ and $\alpha'$ be as before, and let $M$ be a $\sigs$-formula which does not contain a subterm of the form $\lnot F$.
	If $\arel_\boolsort$ is $\le$ and $\arel_\iota$ is $=$ for each $\iota\in\btypes$ and $\Bs\arel\Bs'$ and $\alpha\arel\alpha'$ then $\sinti\Bs M(\alpha)\arel\sinti{\Bs'} M(\alpha')$.
\end{restatable}
\begin{proof}
  Since $\lr_\iota$ is equality, $\alpha'$ must agree with $\alpha$ on variables of type $\iota$ hence $\Bs\sem{M}(\alpha)=\Bs'\sem{M}$ when $M$ is a 1st order $\sigb$ term,  hence $\Bs\sem{M}(\alpha)\lr_o\Bs'\sem{M}$. By \cref{lem:termmonapp}, this also holds when $M$ is a $\sigf$-term. Otherwise, $M=M_1\land M_2$ or $M=M_1\land M_2$ and the property holds by induction on the structure of $M$.
  \end{proof}

\tcb{TODO: Consider making $\arel$ a concrete family of relations.}

\section{Supplementary materials for \cref{sec:ldz}}
\label{apx:sec:ldz}
\begin{example}[Multiplication, Diophantine equations]
	\label{eg:multiplication}
\begin{figure}
{\small
\begin{align*}
	\mathrm{Mul}_1, \mathrm{Mul}_2 &: (\Z \to o) \to (\Z \to o) \to (\Z \to o) \to (\Z \to o) \to \Z \to o\\
	\mathrm{Dec}_1, \mathrm{Dec}_2 &: (\Z \to o) \to (\Z \to o) \to \Z \to o\\
	\mathrm{Inc}_1, \mathrm{Inc}_2 &: (\Z \to o) \to (\Z \to o) \to \Z \to o\\
	\mathrm{Fn} &: \Z \to (\Z \to o)\to o\\
	\mathrm{Gt} &: (\natty\to o)\to (\natty\to o)\to \Z \to o\\
	\mathrm{Gt'} &: (\natty\to o)\to (\natty\to o)\to \natty\to\natty\to \Z \to o\\
	\end{align*}
\begin{align*}
\mathrm{Dec}_1\; f_1\; f_2\; x &\impliedby f_1\; y\land  f_2\; (-y)\land  x\ge y-1\\
\mathrm{Dec}_2\; f_1\; f_2\; x &\impliedby f_1\; y\land  f_2\; (-y)\land  x\ge -(y-1)\\
\mathrm{Inc}_1\; f_1\; f_2\; x &\impliedby f_1\; y\land  f_2\; (-y)\land  x\ge y+1\\
\mathrm{Inc}_2\; f_1\; f_2\; x &\impliedby f_1\; y\land  f_2\; (-y)\land  x\ge -(y+1)\\
\mathrm{Mul}_1\; f_1\; f_2\; g_1\; g_2\; x &\impliedby y=0\land  f_1\; y\land  f_2\; (-y)  x\ge 0\\
\mathrm{Mul}_1\; f_1\; f_2\; g_1\; g_2\; x &\impliedby g_1\; z\land  g_2\; (-z)\ \land \\
         &\mathrm{Mul}_1\; (\mathrm{Dec}_1\; f_1\; f_2)\; (\mathrm{Dec}_2\; f_1\; f_2)\; g_1\; g_2\; w\ \land \\
         &\mathrm{Mul}_2\; (\mathrm{Dec}_1\; f_1\; f_2)\; (\mathrm{Dec}_2\; f_1\; f_2)\; g_1\; g_2\; (-w)\ \land \\
         &x\ge w+z\\
\mathrm{Mul}_1\; f_1\; f_2\; g_1\; g_2\; x &\impliedby g_1\; z\land  g_2\; (-z)\ \land \\
         &\mathrm{Mul}_1\; (\mathrm{Inc}_1\; f_1\; f_2)\; (\mathrm{Inc}_2\; f_1\; f_2)\; g_1\; g_2\; w\ \land \\
         &\mathrm{Mul}_2\; (\mathrm{Inc}_1\; f_1\; f_2)\; (\mathrm{Inc}_2\; f_1\; f_2)\; g_1\; g_2\; (-w)\ \land \\
         &x\ge w-z\\
\mathrm{Mul}_2\; f_1\; f_2\; g_1\; g_2\; x &\impliedby y=0\land  f_1\; y\land  f_2\; (-y) \land x\ge -0\\
\mathrm{Mul}_2\; f_1\; f_2\; g_1\; g_2\; x &\impliedby g_1\; z\land  g_2\; (-z)\ \land \\
         &\mathrm{Mul}_1\; (\mathrm{Dec}_1\; f_1\; f_2)\; (\mathrm{Dec}_2\; f_1\; f_2)\; g_1\; g_2\; w\ \land \\
         &\mathrm{Mul}_2\; (\mathrm{Dec}_1\; f_1\; f_2)\; (\mathrm{Dec}_2\; f_1\; f_2)\; g_1\; g_2\; (-w)\ \land \\
         &x\ge -(w+z)\\
\mathrm{Mul}_2\; f_1\; f_2\; g_1\; g_2\; x &\impliedby g_1\; z\land  g_2\; (-z)\ \land \\
         &\mathrm{Mul}_1\; (\mathrm{Inc}_1\; f_1\; f_2)\; (\mathrm{Inc}_2\; f_1\; f_2)\; g_1\; g_2\; w\ \land \\
         &\mathrm{Mul}_2\; (\mathrm{Inc}_1\; f_1\; f_2)\; (\mathrm{Inc}_2\; f_1\; f_2)\; g_1\; g_2\; (-w)\ \land \\
         &x\ge -(w-z)\\
\mathrm{Fn}\; x\; f &\impliedby f\; y \land x\ge y\\
\mathrm{Gt}\; h_1\; h_2\; x &\impliedby \mathrm{Gt'}\; h_1\;h_2\;\mathrm{C}_{01}\; \mathrm{C}_{01}\; y \land x\ge y\\
\mathrm{Gt'}\; h_1\; h_2\;f_1\;f_2\; x &\impliedby \mathrm{Gt'}\; h_1\;h_2\;(\mathrm{Inc}_1\;f_1\;f_2)\; (\mathrm{Inc}_2\;f_1\;f_2)\; y \land x\ge y\\
\mathrm{Gt'}\; h_1\; h_2\;f_1\;f_2\; x &\impliedby \mathrm{Gt'}\; h_1\;h_2\;(\mathrm{Dec}_1\;f_1\;f_2)\; (\mathrm{Dec}_2\;f_1\;f_2)\; y \land x\ge y\\
\mathrm{Gt'}\; h_1\; h_2\;f_1\;f_2\; x &\impliedby h_1\;f_1\land h_2\;f_2 \land f_1\; y\land f_2\; (-y) \land x\ge y\\
\end{align*}}
\caption{Coding Multiplication \label{fig:multiplication}}
\end{figure}

In \cref{fig:multiplication}, $\mathrm{Inc}_i$ and $\mathrm{Dec}_{i}$ increment or decrement an integer represented by a pair of functions.
Define the following families of formulas:
\[
x:\iota\vdash\mathrm{GT(x)} : \Z\to o
\]
where $\mathrm{GT}(x) := \mathrm{Gt}\; (\mathrm{Fn}\; x)\; (\mathrm{Fn}\; (-x))$, and
\[
x,y,z:\iota\vdash\mathrm{MUL}(x,y,z)
\]
where $\mathrm{MUL}(x,y,z)$ is defined to be
\[
\begin{array}{ll}
& \mathrm{Mul}_1\, (\mathrm{GT}(x))\, (\mathrm{GT}(-x))
\, (\mathrm{GT}(y))\, (\mathrm{GT}(-y))\, z\\
\land &
\mathrm{Mul}_2\, (\mathrm{GT}(x))\, (\mathrm{GT}(-x))\, (\mathrm{GT}(y))\, (\mathrm{GT}(-y))\, (-z)
\end{array}
\]
The family $\mathrm{GT}$ can be used to turn an integer $x$ of type $\Z$ into either function from the pair representing $x$. This allows us to obtain $\mathrm{MUL}$ such that in the canonical model, $\mathrm{MUL}\ x\ y\ z$ holds iff $x \times y=z$.

With this we can define a goal clause corresponding to any Diophantine equation. For example the equation $x^3=y+z$ corresponds to the goal clause \[
 \mathrm{MUL}(x,\,x,\,w) \land \mathrm{MUL}(x,\,w,\,y+z).
\]  
The limit \dz\ problem consisting of this clause together with the set of clauses above are satisfiable if, and only if, the Diophantine equation has a solution.
\end{example}

\section{Supplementary materials for \cref{sec:order-initiality}}

\begin{lemma}
    \label{lem:isframe}
    Given a {$(\sigf,\,\Hf)$}-structure $\Bs$, $\oconfusionf{\Bs}_n$ is a frame.
\end{lemma}
\begin{proof}
To see that $\oconfusionf{\Bs}_n$ meets the requirements on $\boolsort$, $\WW$ and $S$, note that $\Hf$ is a frame, and agrees with it on those types.

The other condition is that for types $\sigma$ of the form $\sigma_1\to\sigma_2$, we need to show that $\oconfusionf{\Bs}_n\sem{\sigma}$ is a subset of $\fspace{\oconfusionf{\Bs}_n\sem{\sigma_1}}{\oconfusionf{\Bs}_n\sem{\sigma_2}}$:

If $\order(\sigma) \le n-2$, or $\order(\sigma)\le n-1$ and $\sigma$ is not an active type, then $\order(\sigma_1)\le n-2$ and $\order(\sigma_2) \le n-2$, or $\order(\sigma_2)\le n-1$ and $\sigma_2$ is not an active type, so $\oconfusionf{\Bs}_n\sem{\tau} = \Hf\sem{\tau}$ where $\tau$ is $\sigma$, $\sigma_1$ or $\sigma_2$, hence the property holds because $\Hf$ is a frame.

If $\order(\sigma)=n-1$ (and $\sigma$ is an active type), then the non-trivial case is when $\sigma=\WW\to \tynint$. Here, we rely on the fact that $\order(\nu)\le n-1$ and $\nu$ is not an active type, hence $\oconfusionf{\Bs}_n\sem{\nu} =\Hf\sem{\nu}$.  For any ${X: \overline{\tau} \to \WW \to\tynint \in \Xwouts}$,  $\assigned{X}{\Bs}$ is an element of $\Hf\sem{\overline{\tau} \to \WW \to\tynint}$ hence $\assigned{X}{\Bs} \, \overline{s}$ is in $\Hf\sem{\WW\to\tynint}\subseteq{\fspace{\WW}{\Hf\sem{\nu}}}$. Therefore  $\oconfusionf{\Bs}_n\sem{\sigma}$ is a subset of $\fspace{\oconfusionf{\Bs}_n\sem{\sigma_1}}{\oconfusionf{\Bs}_n\sem{\sigma_2}}$.

Otherwise, $\oconfusionf{\Bs}_n\sem{\sigma}$ is explicitly constructed as a subset of $\fspace{\oconfusionf{\Bs}_n\sem{\sigma_1}}{\oconfusionf{\Bs}_n\sem{\sigma_2}}$

\end{proof}

\orderfinite*
\begin{proof}
    We prove by lexicographical induction on $n$ followed by the structure of $\sigma$.


    If $\order(\sigma)=0$ then $\Hf_n\sem{\sigma}=\{0,1\}$ or $S$, hence finite.

    Suppose $\order(\sigma)>0$.
    We may assume that either $\order(\sigma) = n-1$ and $\sigma$ is an active type, or $\order(\sigma)=n$ and $\sigma$ is an inactive type, for otherwise 
      $\Hf_n\sem{\sigma} := \Hf_{n-1}\sem{\sigma}$ which is finite by the induction hypothesis.

    If $\sigma=\WW\to \tynint$ then
    \[
    \Hf_n\sem{\sigma} :=\big\{\assigned{X}{\Bs_{n-1}} \, \overline s \mid
        X:\overline \tau \to\WW\to\tynint\in\Swouts,
                             s_i\in\Hf_{n-1}\sem{\tau_i}  \big\}
    \]
    which is finite because there are only finitely many $X\in\Swouts$, and $\Hf_{n-1}\sem{\tau_i}$ is finite (because $\order(\tau_i) < \order(\WW \to \tynint) = n-1$).

    Otherwise, $\sigma=\sigma_1\to\sigma_2$ with $\sigma_1 \not= \WW$.

    If $\sigma$ is an active type of order $n-1$, then $\order(\sigma_1) < n-1$, and so $\Hf_n\sem{\sigma_1} = \Hf_{n-1}\sem{\sigma_1}$ which is finite by IH.
    Now $\sigma_2$ is an active type of order at most $n-1$. By the IH ($\sigma_2$ is smaller than $\sigma$), $\Hf_n\sem{\sigma_2}$ is finite.

    If $\sigma$ is an inactive type of order $n$, then $\order(\sigma_1) \leq n-1$.
    The type $\sigma_1$ could be an active type, but it is smaller than $\sigma$, and so, by IH we have $\Hf_n\sem{\sigma_1}$ is finite.
    The type $\sigma_2$ is inactive and has order at most $n$, but as it is smaller than $\sigma$, $\Hf_n\sem{\sigma_2}$ is finite by IH.

    Therefore, in both cases, $\bnsem{\sigma}=\fspace{\bnsem{\sigma_1}}{\bnsem{\sigma_2}}$ is finite.


    \end{proof}
\newcommand{\Frm}{\textrm{Form}}
\orderonre*
\begin{proof}
    By induction on $\textrm{arity}(\sigma)$.
    If $\sigma=\WW\to\tynint$ and $\order(\sigma)=n$, then $\tynint = \sigma_1\to\cdots\to\sigma_k\to\boolsort$ and $\order(\sigma_i) \leq n-1$.
    It follows from Lemma~\ref{lem:order-finite} that each $\bnsem{\sigma_i}$ is finite.
    Thanks to \cref{ass:ordering} each monotone function $f : \fspace{\WW}{\Hf_n\sem{\tynint}}$ can be described by a map
    \[
    g :\fspace {\bnsem{\sigma_1}}{\cdots\to \fspace{\bnsem{\sigma_k}}{\Frm}\cdots }
    \]
    where $\Frm$ is the set of formulas over $\Sigma_\WW$ with one free variable $x$ and $f(z,\overline{s}) = \As\sem{g(\overline{s})}([x\mapsto z])$.
    The set of such $g$ is recursively enumerable ($\Frm$ is r.e. and $\Hf\sem{\sigma_i}$ are finite).

    Otherwise, $\sigma=\sigma_1 \to \tyint$ with $\order(\sigma_1) < n$, so $\Hf_n\sem{\sigma_1}$ is finite, hence $\fspace{\Hf_n\sem{\sigma_1}}{\Hf_n\sem{\tyint}}$ is r.e.~if $\Hf_n\sem{\tyint}$ is, which holds by induction on arity.
    \end{proof}

\orderre*
\begin{proof}
    For $i>l+1$, $\Bs_i$ is a structure extending $\Bs_{i-1}$ over the same signature,
    therefore $\Bs_i=\Bs_{l+1}$ by induction. Hence we only need to show that there are enumerably  many possible $\Bs_{l+1}$.

    There is exactly one possible $\Bs_0$.

    If we fix $\Bs_{n-1}$, then to construct the structure $\Bs_n$ we need to choose interpretations for
    predicates in $\Swouts$ which have order $n$ type, and predicates that are inactive type with order $n-1$.
    Since these choices must come from r.e.~or finite sets (by Lemmas~\ref{lem:order-lvnre} and~\ref{lem:order-finite}), there are enumerably many options for $\Bs_n$.
    Since an r.e.~collection of r.e.~sets is r.e., the number of families (up to order $l+1$) is r.e. by induction.
    \end{proof}

\orderemtonore*
\begin{proof}
For each resolution proof rule, if $\Bs_{l+1}$ entails the premises of the rule, then it entails the conclusion. Since $\Bs_{l+1}\sem{\bot}=0$, there is no resolution proof of $\bot$.

\sr{8/1: Changed occurrences of subscript $m$ to $l+1$.}

Resolution:

\tcb{These inferences rely on properties of all frames, and are not specific to \oconfused ones. This may be stated in \cite{DBLP:conf/lics/OngW19} or in the unpublished one.}
If $\Bs_{l+1} \models {\neg R\,\overline M\lor G}$ and $\Bs_{l+1}\models {G'\lor R\,\overline x}$, then for each $\alpha$ either $\Bs_{l+1}\sem{\;R\ \overline{M}}(\alpha)=1$ (and $\Bs_{l+1}\sem{G}(\alpha)=1$) or $\Bs_{l+1}\sem{R\ \overline{M}}(\alpha)=0$ (and we have $\Bs_{l+1}\sem{(G'[\overline M/\overline x]\big)}(\alpha)=\Bs_{l+1}\sem{G'}(\alpha[\overline M/\overline x])=1$). Therefore $\Bs_{l+1}\sem{G\lor \big(G'[\overline M/\overline x]\big)}(\alpha)=1$

Constraint refutation:

If there exists a valuation $\alpha$ such that $\As,\alpha\models\phi_1\land\cdots\land\phi_n$ then, setting $\alpha(x_i)=\top$ for each $i$, $\Bs_{l+1}\sem{\lnot (x_1\,\overline M_1)\lor\cdots\lor \lnot(x_m\,\overline M_m)\lor\neg\phi_1\lor\cdots\lor\neg\phi_n}(\alpha)=0$, so if $\Bs_{l+1} \models \lnot (x_1\,\overline M_1)\lor\cdots\lor \lnot(x_m\,\overline M_m)\lor\neg\phi_1\lor\cdots\lor\neg\phi_n$, we get $\Bs_{l+1} \models \bot$ vacuously.
\tcb{This does rely on the presence of top elements in each type.}
\end{proof}

\ordersmtoem*
\begin{proof}
	%
	%
	%
	
	The proof is by logical relations and \cref{cor:lrposterm}.
	
	Given a standard model $\Bs$, we construct a family of \oconfused structures by induction on $n$. As we construct an \oconfused family $\{\Bs_n\}_{n \in \omega}$,
	we define logical relations $\arel^n$ between $\Hf_n := \oconfusionf{\Bs_{n-1}}_n$ and $\Sf$ recursively as follows:
	\begin{align*}
		i\arel_{\iota}^n i' & \defeq i=i'\qquad \iota \in \btypes\\ 
		b\arel_o^n b' & \defeq b\leq b'\\ 
		r\arel_{\tau\to\sigma}^n r' & \defeq \forall s\in\Hf_n\tden\tau,s'\in\Sf\tden\tau \ldotp
		s\arel_\tau^n
		s'\rightarrow r(s)\arel_\sigma^n r'(s')
	\end{align*}
	Thus $\arel^n$ is, by construction, the unique logical relation which is $\le$ on $\boolsort$ and is $=$ on other base types.
	
	Let us suppose the $(\Sigma_{n-1},\Hf_{n-1})$-structure $\Bs_{n-1}$ is defined.
	We construct $\Bs_{n}$ as a $(\Sigma_{n},\Hf_{n})$-expansion of $\Bs_{n-1}$ by:
	\[
	\assigned{Y}{\Bs_n}f_1 \cdots f_r := \bigwedge_{\overline{f}\lr\overline{g}}\assigned{Y}{\Bs}g_1 \cdots g_r
	\]
	where $Y:\sigma_1\to\cdots\to\sigma_r\to\boolsort$ is an order-$n$ initial type, and each $f_i \in \Hf_n\sem{\sigma_i}$.
\paragraph*{Showing that $\Bs_n$ is well defined}
	We need to show that $\assigned{Y}{\Bs_n} \in \Hf_n \sem{\overline \sigma \to o}$.
	Suppose $Y$ is of an order-$n$ active type $Y : \overline{\tau} \to \WW \to \overline{\alpha} \to \boolsort$ where:
	\begin{itemize}
		\item each $\tau_i$ (and each $\alpha_j$) is initial or $S$ 
		\item for each $i$, $\order(\tau_i) < \order(\WW \to \overline{\alpha} \to \boolsort)$.
	\end{itemize}
	Since $\order(\WW \to (\overline\alpha \to \boolsort)) = n$ and:
	\[
	\begin{array}{ll}
		& \Hf_n\sem{\overline{\tau} \to \WW \to \overline{\alpha} \to \boolsort} \\
		= & \big[\Hf_n\sem{\tau_1} \to \cdots \to [\Hf_n\sem{\tau_k} \to \Hf_n\sem{\WW \to (\overline\alpha \to \boolsort)}]\big]
	\end{array}
	\]
	we need to show that, for each $s_i \in \Hf_n\sem{\tau_i}$, $\assigned{Y}{\Bs_n} \; \overline s$ is monotone in the first argument.
	I.e., assume $z \leq w$, we want to show $\assigned{Y}{\Bs_n} \; \overline s \; z \leq_o \assigned{Y}{\Bs_n} \; \overline s \; w$;
	or equivalently
	\begin{align}
		\forall f_1, \cdots, f_k \, . \,
		\assigned{Y}{\Bs_n} \; \overline s \; z \; f_1 \cdots f_k \leq \assigned{Y}{\Bs_n} \; \overline s \; w \; f_1 \cdots f_k.
		\label{eq:order-smtoen-monotone}
	\end{align}
	Now, since $\Bs$ is a model of $\Gamma$ which is assumed to be a limit Datalog problem, $\assigned{Y}{\Bs} \; \overline t \in \Sf\sem{\WW \to \overline{\alpha} \to \boolsort}$ is upward closed in the first argument for all $\overline t$ (and in particular for those satisfying $\overline s \lr \overline t$), i.e., for all $\overline g$ (and in particular for those satisfying $\overline f \lr \overline g$)
	\[
	\assigned{Y}{\Bs} \; \overline t \; z \; \overline g \leq \assigned{Y}{\Bs} \; \overline t \; w \; \overline g
	\]
	which implies~(\ref{eq:order-smtoen-monotone}), by considering the definition of $\assigned{Y}{\Bs_n}$.
	
	Next suppose $Y:\sigma_1\to\cdots\to\sigma_k\to\boolsort$ is an inactive type.
	Then
	\[
	\Hf_n\sem{\overline \sigma \to \boolsort} =
	\big[\Hf_n\sem{\sigma_1} \to \cdots \to [\Hf_n\sem{\sigma_k} \to \Hf_n\sem{\boolsort}]\big]
	\]
	(and each $\Hf_n\sem{\sigma_i}$ is finite, by \cref{lem:order-finite}).
	It follows from the definition that $\assigned{Y}{\Bs_n} \in \Hf_n\sem{\overline\sigma \to o}$, as desired.
	
	Note that $\Bs_{l+1}\lr \Bs$. The construction of $\assigned{Y}{\Bs_n}$ as a greatest lower bound matches the definition of logical relations and ensures that $\Bs_{l+1}$ is the greatest \oconfused model such that this holds.
	
	
	\paragraph*{We now show that $\Bs_{l+1}$ models $\Gamma$}
	
	Consider a goal clause $G\in\Gamma$. For any $(\Delta, \Hf_{l+1})$-valuation $\alpha$, there is a standard valuation $\alpha'$ such that $\alpha'\glr\alpha$
	(just take $\alpha'(x)=\alpha(x)$ on variables of type $\WW$ or $S$ and take $\alpha'(x)=\top$ otherwise).
	Since $\Bs$ is a model, $\Bs\sem{G}(\alpha')=1$, and by \cref{lem:termmonapp}
	, $\Bs_{l+1}\sem{G}(\alpha)=1$. Hence $G$ is satisfied by $\Bs_{l+1}$.
	
	Now consider a definite clause $X\,\overline{x} \vee G \in \Gamma$ and some {$(\Delta, \Hf_{l+1})$-valuation} $\alpha$.
	If $\Bs_{l+1}\sem{X\ \overline{x}}(\alpha)=0$ then we have $\assigned{X}{\Bs_{l+1}}\alpha(x_1)\cdots \alpha(x_k)=0$.
	By construction of $\Bs_{l+1}$, there exist $g_1,\cdots, g_k$ such that $\assigned{X}{\Bs} \, \overline g=0$ and each $g_i\glr\alpha(x_i) $.
	This allows us to construct $\alpha'\glr \alpha$
	$\Bs\sem{X\ \overline{x}}(\alpha')=0$.
	Again, since $\Bs$ is a model, $\Bs\sem{G}(\alpha)
	= 1$, and by  \cref{cor:lrposterm}
	, $\Bs_{l+1}\sem{G}(\alpha)= 1$.
	Therefore $\Bs_{l+1}$ satisfies $X\,\overline{x} \vee G$.
\end{proof}
\orderdecidablemodel*
\begin{proof}
	We prove this by converting the clause into a first-order formula over $\Sigma_\WW$.
	We obtain the formula by a 2 step transformation. These steps involve formulas of higher-order logic as defined in \cref{sec:prelims}, not just those with the structure of Horn clauses.

	\subparagraph{Preprocessing - replacing variables (not of type $\WW$) by constants}
	We begin by assuming that for each relational type $\rho$, for each element of $f\in \Bs_{l+1}\sem{\rho}$, $\sigf$ contains a predicate symbol $X_f$.
	We also assume that $\assigned{\Bs}{X_f} = f$.
	If this did not hold, it is trivial to construct a new $\sigf$ and the corresponding $\Bs_{l+1}$.
	This obviously does not affect satisfiability of a set of clauses which do not mention the newly added predicate symbols.
	
	For each variable $x$ in $G$ that is not of type $\WW$, $\Bs_{l+1}\sem{\Delta(x)}$ is finite.
	Therefore we may replace $G$ by a conjunction of clauses -- one for each possible interpretation of $x$ (we introduce constants corresponding to each possibility).
	
	Every foreground atom in the clause now has the shape $X\,\overline{M}$ where at most one $M_i$ is a variable (of type $\WW$), thanks to the initiality restriction on the type of $X$. All other $M_i$ are either constants or have this same shape.

	\subparagraph{Eliminating variables in foreground terms}
	We will now eliminate variables of type $\WW$ in foreground atoms. We do so by induction on the number of free variables in a foreground atom - at each step, we replace an atom by a disjunction of guarded atoms, each with one fewer free variable.
	
	A foreground atom has the form $X\,\overline{M}$ and not more than one of $M_i$ are variables (because all remaining variables are of type $\WW$ and $X$ has initial type). If exactly one $M_j$ is a variable $x$ and no other $M_i$ contain variables, then we may replace the atom $X\,\overline{M}$ by the background formula $\phi(x)$ corresponding to the upset $\{w\in\WW\ |\ \Bs\sem{X\,\overline{M}}([x\mapsto w])\}$. Such a formula exists by \cref{ass:ordering} (see also the proof of \cref{lem:order-lvnre}, which shows we can enumerate entwined structures while having access to the formulas corresponding to these upsets).
	
	For foreground atoms involving a more deeply nested variable (including those with more than one variable), any atom $M$ with at least one variable must contain some subterm $N$ of the shape $X\,\overline{M}\,x\,\overline{N}$ where $x$ is a variable and neither $\overline{M}$ nor $\overline{N}$ contain any variables. \tcb{What should I call $M[N]$? Do I need to define it better?} We take $M[-]$ to be the one-holed context such that $M[N]=M$. Take the type of $X$ to be $\overline{\sigma}\to\WW\to\overline{\sigma'}\to\overline{\tau}\to \boolsort$ where $\overline{\sigma'}$ are the types of $\overline{N}$. Note that the type of $N$ is $\overline{\tau}\to\boolsort$.
	
	Since $X$ has an initial type, we know that each $\Bs\sem{\tau_i}$ is finite. Therefore the set of tuples $T = \Pi_{i=1}^n\Bs\sem{\tau_i}$ is finite. For each tuple $\overline{t}\in T$, there is a formula $\phi_{\overline{t}}(x)$ (like $\phi$ above) defining the values $z$ such that $\Bs\sem{N\, \bar{y}}([x\mapsto z,\overline{y}\mapsto\overline{t}])=1$. If we now consider $S\subseteq T$, the function $f(z) = \Bs\sem{N}([x\mapsto z])$ (which returns a function in $\fspace{T}{\mathbb{B}}$) is constant on the regions where $\bigwedge_{\overline{t}\in S}\phi_{\overline{t}}(x)\land \bigwedge_{\overline{t}\in T\setminus S}\lnot\phi_{\overline{t}}(x)$ holds (if the formula holds at $x=z$ then $f(z)(\overline{t})$ can be determined by whether or not $\overline{t}\in S$, hence does not depend on the precise value of $z$). Denote this constant by $N_S$.
	
	Since $\WW$ can be partitioned by according to the powerset of such $\phi_{\overline{t}}$, we may replace $M$ by the formula $\bigvee_{S\subseteq T}(\bigwedge_{\overline{t}\in S}\phi_{\overline{t}}(x)\land \bigwedge_{\overline{t}\in T\setminus S}\lnot\phi_{\overline{t}}(x)\land M[N_S])$. Since $N_S$ contains no variables, each $M[N_S]$ contains one fewer variable than $M$, allowing us to inductively remove all variables from atoms until we are left with constant expressions that reduce to booleans and background atoms. Since the first-order theory of $\Sigma_\WW$ is decidable, we are done.\tcb{This needs a bit more work in parts, but the proof is there}
\end{proof}

\subsection{Example of an Entwined Structure}
\label{sec:entwined-example}

Let the background theory be the theory of equality on the finite set $S = \{\clubsuit, \diamondsuit, \spadesuit, \heartsuit\}$ in combination with the theory of linear integer arithmetic, in which $\Z$ is ordered by $\leq$.
We give an entwined structure that interprets the three predicate symbols $X$, $Y$ and $Z$ from Example~\ref{eg:xyz}.  
The signature of the foreground, $\sigf$, in this case consists of:
\begin{align*}
  X &: \rho\\
  Y &: S\to S\to\boolsort\\
  Z &: \WW \to \rho \to\boolsort
\end{align*}
where $\rho$ is shorthand for $S\to (S\to\boolsort) \to\boolsort\to \xi$ and $\xi$ for $\WW \to (\WW \to\boolsort) \to\boolsort$.
So the type of $Y$ is order 1, the type of $X$ order 2 and the type of $Z$ order 3.
Correspondingly, we have $\Sigma_1 = \{Y : S\to S\to\boolsort\}$, $\Sigma_2 = \{X:\rho,\,Y : S\to S\to\boolsort\}$ and then $\Sigma_3 = \sigf$.

We can build an entwined structure to interpret $\sigf$ in stages, according to the definition.
\begin{itemize}
  \item Define $\Bs_0$ as the unique $(\emptyset,\,\Sf)$-structure, i.e. that interprets the background theory standardly and interprets the empty signature vacuously.
  In particular, we have the following interpretations of the base types:
    \begin{align*}
      \Sf\tden{o} &= \mathbb{B} \quad (= \{0,1\}) \\
      \Sf\tden{S} &= \{\clubsuit, \diamondsuit, \spadesuit, \heartsuit\} \\
      \Sf\tden{\WW} &= \mathbb{Z}
    \end{align*}

  \item Then $\oconfusionf{\Bs_0}_1$ the order-1 entwined frame derived from $\Bs_0$ is determined by the definition.  According to clause (i), the base types are interpreted as in $\Sf$ and by clause (iii) we have, in particular, the following interpretations of first-order types:
  \begin{align*}
    \oconfusionf{\Bs_0}_1\tden{S \to o} &= [S \to \mathbb{B}] \\
    \oconfusionf{\Bs_0}_1\tden{S \to (S \to o)} &= [S \to [S \to \mathbb{B}]]
  \end{align*}
  and $\oconfusionf{\Bs_0}_1\tden{W \to o}$ as the set of monotone functions:
  \[
     \{f \in [\Z \to \mathbb{B}] \mid \forall z \leq z'.\, f(z) \leq f(z')\}
  \]
  Of course, all other types are interpreted too, but these are the order-1 types that will be important in assigning a meaning to $X$, $Y$ and $Z$.

  \item We can use $\oconfusionf{\Bs_0}_1$ to frame the entwined interpretation of the order-1 foreground symbols in $\Sigma_1$, via a $(\Sigma_1,\oconfusionf{\Bs_0}_1)$-expansion of $\Bs_0$.  The definition of expansion forces $\Bs_1$ to interpret the background theory in the same way as $\Bs_0$ (i.e. standardly) but we have free choice of interpretation of $Y \in \Sigma_1$ as any element of $[S \to [S \to \mathbb{B}]]$.  Let us pick (it is not important here since we are not interested in satisfying a particular set of clauses):
  \[
    \Bs_1\tden{Y} = 
      \begin{cases}
        1 & \text{if $\{s_1,\,s_2\} \subseteq \{\diamondsuit,\spadesuit\}$}\\
        0 & \text{otherwise}
      \end{cases}
  \]

  \item Then $\oconfusionf{\Bs_1}_2$, the order-2 entwined frame derived from $\Bs_1$ is determined according to the definition.  From clause (i) we have that the base types are interpreted as in $\oconfusionf{\Bs_0}_1$ and also the inactive initial types:
  \begin{align*}
    \oconfusionf{\Bs_1}_2\tden{S \to o} &= \oconfusionf{\Bs_0}_1\tden{S \to o} \\
    \oconfusionf{\Bs_1}_2\tden{S \to S \to o} &= \oconfusionf{\Bs_0}_1\tden{S \to S \to o} \\
  \end{align*}
  In this way, the interpretation $\Bs_1\tden{Y}$ still makes sense within this frame.
  According to clause (ii), the active initial type $\WW \to o$ is reinterpreted as follows:
  \[
    \oconfusionf{\Bs_1}_2\tden{\WW \to o} = \{\top\}
  \]
  since there are no terms of type $\WW \to o$ that can be obtained as partial applications of $X$, $Y$ or $Z$.
  By clause (iii), we have:
  \[
    \oconfusionf{\Bs_1}_2\tden{\rho} = [S \to [[S \to \mathbb{B}] \to [\mathbb{B} \to \oconfusionf{\Bs_1}_2\tden{\xi}]]]
  \]
  and $\oconfusionf{\Bs_1}_2\tden{\xi}$ is the set: 
  \[
    \{ f \in [\Z \to [\{\top\} \to \mathbb{B}]] \mid \forall z \leq z'.\, f(z) \leq f(z')\}
  \]

  \item When we define $\Bs_2$ as a $(\Sigma_2,\oconfusionf{\Bs_1}_2)$-expansion of $\Bs_1$ we are forced by the notion of expansion to take the interpretation of the background as in $\Bs_1$ and also:
  \[
    \Bs_2\tden{Y} = \Bs_1\tden{Y}
  \]
  The definition of entwined frame ensures that the interpretation also makes sense in $\oconfusionf{\Bs_1}_2\tden{S \to S \to o}$.
  On the other hand, we are free to choose any element of the infinite set $\oconfusionf{\Bs_1}_2\tden{\rho}$ with which to interpret the other element of $\Sigma_2$, namely $X$.  We pick:
  \[
    \Bs_2\tden{X}(s)(f)(b)(w)(g) = 
      \begin{cases}
        1 & \text{if $b$ and $f(s)$ and $w>5$}\\
        0 & \text{otherwise}
      \end{cases}
  \]
  Here we can see concretely the intuition explained in Example~\ref{eg:xyz}: even though $Z\,5\,X$ is a potential (third) argument to $X$ and so must be accounted for when describing how to interpret the type of $X$ (so that the application is defined), we can understand the (finitely many) values that are possible for $Z\,5\,X$ without knowing how to interpret $Z$.

  \item Then $\oconfusionf{\Bs_2}_3$ is determined as follows.  By (i) all base types, the inactive initial types of order-1 and, now, also the active initial type $\WW \to o$ are interpreted as in the previous entwined frame.  Next, by (ii) we have:
  \[
    \oconfusionf{\Bs_2}_3\tden{\rho} = [S \to [[S \to \mathbb{B}] \to [\mathbb{B} \to \oconfusionf{\Bs_2}_3\tden{\xi}]]]
  \]
  and $\oconfusionf{\Bs_2}_3\tden{\xi}$ is the set: 
  \[
    \{\top\} \cup \{\Bs_2\tden{X}(s)(f)(b) \mid s \in S \wedge f \in \fspace{S}{\boolsort} \wedge b \in \mathbb{B}\}
  \]
  which is equal to $\{\top,\bot, (w\mapsto (g\mapsto w>5))\}$.

  By clause (iii) we interpret the order-3 type $\WW \to \rho \to o$ as:
  \[
    \{f \in [\Z \to [\oconfusionf{\Bs_2}_3\tden{\rho} \to \mathbb{B}]] \mid \forall z \leq z'.\, f(z) \leq f(z') \}
  \]

  \item Now we can define a $(\Sigma_3,\oconfusionf{\Bs_2}_3)$-expansion of $\Bs_2$.  We are forced to take $\Bs_3\tden{Y} = \Bs_2\tden{Y}$ and $\Bs_3\tden{X} = \Bs_3\tden{X}$, but this is possible because the definition of $\oconfusionf{\Bs_2}_3\tden{\rho}$ ensures that $\Bs_2\tden{X}$ remains an element.  We are free to choose an appropriate way to interpret the third-order symbol $Z$, let us take:
  \[
    \Bs_3\tden{Z}(w)(f) = 
      \begin{cases}
        1 & \text{if $f(\heartsuit)(\Bs_3\tden{Y}(\spadesuit))(1)(w)(\top)$}\\
        0 & \text{otherwise}
      \end{cases}
  \]
\end{itemize}

\newpage


\iffalse
\section*{Todo List}
\subsection*{Preliminaries}
\begin{list}{-}{}
    \item \sout{Get rid of existentials and $\lambda$-abstractions}
    \item \changed[lo]{\sout{Introduce resolution and mention refutational completeness}}
    \item Check \cref{rem:isimple} (I think we need to explicitly say that negations a)
    \item \changed[lo]{Clarify frame and Comprehension Axiom (restricted to positive existentials) -- reviewers will likely query this.}
\end{list}
\subsection*{Main proof}
\begin{list}{-}{}
    \item Change main definition
    \item Give more detail about decidability of things being models - I think this is more complicated than I thought and requires another lemma
\end{list}

\subsection*{Conjectures}
\changed[lo]{
\begin{list}{-}{}
    \item Extend current result in Section~\ref{sec:ildz} (Theorem~\ref{thm:main-sat-result}) to the more general notion of initial types $\sigma_1 \to \cdots \to \sigma_n \to \boolsort$, where at most one $\sigma_i$'s is $\Z$, and if $\sigma_i = \Z$ then for all $j < i$, $\level(\sigma_j) < \level(\sigma_i \to \cdots \to \sigma_n \to \boolsort)$.

    \item Introduce notion of \emph{order-initial} types, and prove decidability of satisfiability for order-initial limit \dz\ problems.

    \item Prove the conjectures in Section~\ref{sec:naturals}.
\end{list}}

\subsection*{Others}
\begin{list}{-}{}
    \item Abstract
    \item Introduction
    \item \changed[lo]{Examples (also need to settle syntax of example programs / problems in Section~\ref{sec:ldz})}
    \item Related work
    \item Conclusion
\end{list}
\else
\fi



\end{document}


\ifdraft

\section{Other possible type restrictions}
\label{sec:extensions}

\lo{4/1: It may be worth dropping the result about level-initial case, but it and the possibility of extension as a remark. If so, this will save space, and we can just replace order-initial by initial.}

There are two restrictions on HoCHC with a background theory of LIA we impose to obtain decidability. One of these is the ``limit'' requirement - interpretations must be upward closed when seen as sets of integers. The other is the ``initial'' restriction, which restricts the types allowed for predicates. One question to consider is ``What is a maximal set of types such that the higher-order limit \dz\ problem is decidable when predicates are restricted to those types?''

We have not given a maximal set of types.
Note that the hereditary requirement --- (O3) in \cref{def:order-initial} --- of order-initial type is unnecessary, since if we have no predicates with types ending in (for example) $\Z\to\Z\to\boolsort$, there is very little we can usefully say about predicates with type $(\Z\to\Z\to\boolsort)\to\boolsort$.
The hereditary restriction however makes some parts of the proof simpler without altering the expressivity of the language.

\subsection{Level-initial limit \dz}
\label{sec:ildz}

Another way to ensure that the restriction is satisfied is to insist that in the arguments to a predicate, an argument of type $\Z$ must occur before any argument whose type contains an occurrence of $\Z$.
For example, we would admit predicates of type $S \to \Z \to \boolsort$ 
and $\Z \to S \to (\Z \to \boolsort) \to \boolsort$, but not predicates of type $\Z \to \Z \to \boolsort$ nor $(\Z \to \boolsort) \to \Z \to \boolsort$.

\begin{definition}\rm
\label{def:level-initial}
We say that a relational type $\sigma = \sigma_1 \to \cdots \to \sigma_n \to o$ where $n \geq 0$ is \emph{level-initial} if for all $1 \leq i \leq n$
\begin{itemize}
\item[(I1)] if $\sigma_i = \Z$ then for all $j < i$, $\Z$ does not occur in  $\sigma_j$ (it follows that at most one of $\sigma_1, \cdots, \sigma_n$ is $\Z$), and
\item[(I2)] if $\sigma_i \not= S$ and $\sigma_i \not= \Z$, then $\sigma_i$ is level-initial.
\end{itemize}

An \firstdef{level-initial limit \dz\ problem} is a limit \dz\ problem where for each $X : \rho \in \Swouts$, it is the case that $\rho$ is level-initial.
\end{definition}

\tcb{The hereditary part doesn't matter, since if $X$ has type $\sigma\to\tau$ where no symbol has a type ending with $\sigma$, then no term involving $X$ could possibly matter in a resolution proof. Therefore such a symbol is irrelevant to the satisfiability of a problem.}

\lo{15/12: The notion of level-initial types \emph{less} the hereditary requirement (I2) is simpler / cleaner (and easier to describe) and therefore more attractive.
However the proof of Theorem~\ref{thm:main-sat-result} does assume $\tyint$- or $\tynint$-types (and requiring (I1)).}

\lo{15/12: QUESTION. Does Theorem~\ref{thm:main-sat-result} hold for level-initial types less the hereditary requirement (I1)?}

\tcb{ANSWER: Yes, but this relies on the fact that it would be completely useless. I might argue that the problem is more natural with the hereditary constraint}
\tcb{Say that a type is bad if it satisfies the constraint but not hereditarily. The fact that resolution works means that only concrete terms matter (that is the idea behind the whole decidability proof). Consider a resolution proof involving a predicate $P$ with a bad type e.g $(\Z\to\Z\to\boolsort)\to\boolsort$. $P$ cannot be applied to anything (unless we allow existential quantifiers over higher order types in which case it can be applied to exactly one thing, top). Other bad predicates may take $P$ as an argument, but since they can't apply it to anything, they can't do anything meaningful with its behaviour.}


\lo{TODO: We should give examples of ILDP to illustrate its expressivity.

Examples~\ref{eg:iter} and \ref{eg:paths} are ILDP, but Examples~\ref{eg:addition} and \ref{eg:multiplication} are not.}

\begin{theorem}
\label{thm:main-sat-result}
Satisfiability of level-initial limit-\dz\ problems is decidable.
\end{theorem}


The prefix ``level-'' refers to a measure of types similar to (but different from)  order, which we will use to construct an inductive proof for Theorem~\ref{thm:main-sat-result}), defined as follows.
\begin{align*}
\level(o)&:=0\\
\level(S\to\rho)&:=\level(\rho)\\
\level(\Z\to\rho)&:=\level(\rho)+1\\
\level(\rho_1\to\rho_2)&:=\max(\level(\rho_1),\level(\rho_2)) \quad \textrm{if }\rho_1 \not= S, \rho_1 \not= \Z
\end{align*}

\begin{example} 
Except for the first (which is a $\tyzero$-type) and the {third} (which is a $\tynint$-type), all of the following are $\tyint$-types.
\begin{enumerate}
\item $\level((S\to \boolsort)\to S \to \boolsort) =0$
\item $\level(\Z \to \boolsort) = 1$
\item $\level((\Z \to \boolsort) \to \boolsort) = 1$
\item $\level((\boolsort\to \boolsort) \to \Z \to S \to \boolsort) = 1$
\item $\level(\Z \to (\Z \to \boolsort) \to (\Z\to \boolsort) \to \boolsort) = 2$
\item $\level(\Z \to ((\Z \to o)\to o) \to (\Z \to o) \to o) = 2$
\item $\level(\Z \to (\Z\to (\Z \to o)\to o) \to o) = 3$
\end{enumerate}
\end{example}

\paragraph{Proof of \cref{thm:main-sat-result}}
The proof follows essentially the same argument as \cref{thm:order-main-sat-result}.
The details are in \cref{apx:sec:extensions}.
\lo{TODO: Point out the differences. E.g.~use of level to build level-entwined families.}

\subsection{A further extension}

The most obvious significant expansion is to allow types of level $n-1$ in front of $\Z$ in integer types. This would look something like ``types of the form $\sigma_1\to\cdots\to\sigma_n\to\boolsort$ where at most one of $\sigma_1,\cdots,\sigma_n$ is $\Z$ and if $\sigma_i=\Z$ then for all $j<i$, $\level(\sigma_j)<\level(\sigma_i\to\sigma_{i+1}\to\cdots\to\sigma_n\to\boolsort)$'' (this would require a minor change to the definition of level, since the current one only works with the current definition of level-initial).

This preserves the property that we can't apply a predicate to itself (or any partially applied form of itself) unless we already know that there are finitely many distinguishable values that form could have.


\begin{remark}
Level-initial types are neither strictly larger nor strictly smaller than order-initial types. 
\begin{itemize}
\item $((S\to o)\to o)\to \Z\to (\Z\to o)\to o$ is level-initial, but not order-initial; and 
\item $(\Z\to o)\to \Z\to ((S\to o)\to o)\to o$ is order-initial, but not level-initial.
\end{itemize}
\end{remark}

\lo{TODO: Need to rename integer / inactive types to avoid confusion with their use in the preceding section on order-initial types.}

\tcb{TODO: consider adding $S$ to $\tyzero$ to make the shape of $\tyint$ more consistent}

There is an equivalent presentation of level-initial types.

A \emph{level-initial limit \dz\ problem} is a limit \dz\ problem where for each $X: \rho \in \Swouts$, we have that $\rho$ is either a $\tynint$-type (called \emph{inactive type}) or a $\tyint$-type (called \emph{active type}) as defined by the BNF grammar:
\[
\begin{array}{rccl}
\hbox{\emph{Zero type}} & \tyzero &::=& \boolsort \mid S\to\tyzero \mid \tyzero\to\tyzero\\
\hbox{\emph{Inactive type}} & \tynint &::=& o \mid S\to\tynint \mid \tynint\to\tynint \mid \tyint\to\tynint\\
\hbox{\emph{Active type}} & \tyint &::=& \Z\to\tynint \mid S\to\tyint \mid \tyzero\to\tyint
\end{array}
\]
\medskip

Some immediate observations about $\tyint$-, $\tynint$- and $\tyzero$-types.
\begin{itemize}
    \item A zero type is by construction an inactive type (but not vice versa).
    We call $\tyzero$-types zero types because (as we will see) they are precisely the relational types that have level 0.

    \item Every active type $\tyint$ has the shape ($n, m \geq 0$)
    \[
    \tyint =\tyzero_1\to\cdots\to\tyzero_n\to\Z\to\tynint_1\to\cdots\to\tynint_m\to\boolsort
    \]
    \lo{The preceding is not correct: the $\tyzero_i$ and $\tynint_i$ could be $S$!}


    \item Every inactive type $\tynint$ has the shape ($n \geq 0$)
    \[
\tynint = \sigma_1 \to \cdots \to \sigma_n \to \boolsort
    \]
    where each $\sigma_i$ is either $S$, or a $\tynint$-type, or a $\tyint$-type.
\end{itemize}

We give a breakdown on how level works on zero types, inactive types and active types.
\begin{itemize}
\item The level of a zero type is 0.

\item For any inactive type $\tynint = \sigma_1\to\cdots\to\sigma_n$, we have
\[
\level(\tynint) = \max\big(\level(\sigma_1),\cdots,\level(\sigma_m)\big)
\]

\item For any active type $\tyint =\overline{\tau} \to \Z \to \tau_1'\to\cdots\to\tau_m'\to\boolsort$ (where $\tau_i$ is either a zero type or $S$, and each $\tau_j'$ is either a inactive type or $S$),
we have
\[
\level(\tyint) = 1 + \max\big(\level(\tau_1'),\cdots,\level(\tau_m')\big).
\]
\end{itemize}

\noindent\textbf{Notation}.
\emph{In this subsection, whenever we write a type as $\tyint$, we mean that it is an active type; similarly for $\tynint$ and $\tyzero$.}

%


\begin{lemma}
\label{lem:characteriselevelinitial}
For all relational type $\sigma$, $\init(\sigma)$
iff $\sigma$ is a $\tynint$-type or a $\tyint$-type.
\end{lemma}

\begin{proof} 
Let $\sigma = \sigma_1 \to \cdots \to \sigma_n \to o$.
In case $n \geq 1$, we write $\sigma = \sigma_1 \to \sigma'$.

$\Rightarrow$: Assume $\init(\sigma)$.
We prove by structural induction that $\sigma$ is a $\tynint$- or $\tyint$-type.
The base case: $\sigma = o$ is a $\tynint$-type by definition.
For the inductive case, it follows from the definition that $\init(\sigma')$ holds.
Hence, by the IH, $\sigma'$ is either a $\tyint$- or a $\tynint$-type.

Suppose $\sigma'$ is a $\tyint$-type. Then, because of (I1) of $\init(\sigma)$, $\Z$ does not occur in $\sigma_1$, and so, $\sigma_1$ is a $\tyzero$-type.
Hence $\sigma$ is a $\tyint$-type.

Suppose $\sigma'$ is a $\tynint$-type.
If $\sigma_1 = \Z$ or $S$, then $\sigma$ is a $\tyint$-type or $\tynint$-type respectively.
Otherwise, (I2) of $\init(\sigma)$ says that $\init(\sigma_1)$ holds.
By the IH, $\sigma_1$ is a $\tyint$- or $\tynint$-type.
Hence $\sigma$ is a $\tynint$-type.

$\Leftarrow$: Suppose $\sigma$ is a $\tyint$-type.
We show that (I1) and (I2) hold for all $i$.
\begin{itemize}
\item (I1): We have $\sigma'$ is either a $\tynint$- or a $\tyint$-type. By the IH, for $2 \leq i \leq n$, if $\sigma_i = \Z$ then for $2 \leq j < i$, $\Z$ does not occur in $\sigma_j$.
If $\sigma'$ is a $\tyint$-type, then $\sigma_1$ is $S$ or a $\tyzero$-type, i.e., $\Z$ does not occur in $\sigma_1$.
If $\sigma'$ is a $\tynint$-type, then $\sigma_1 = \Z$.

\item (I2): By the IH, (I2) holds for $2 \leq i \leq n$.
We have $\sigma_1$ is $S$ or $\Z$ or a $\tyzero$-type.
But if $\sigma_1$ is a $\tyzero$-type, then $\sigma_1$ is a $\tynint$-type, and so $\init(\sigma_1)$ by the IH.
\end{itemize}

Suppose $\sigma$ is a $\tynint$-type. Then for $1 \leq i \leq n$, $\sigma_i \not= \Z$, i.e., (I1) holds vacuously. As for (I2), observe that $\sigma'$ must be a $\tynint$-type.
Hence, by the IH $\init(\sigma')$, and so, (I2) holds for $2 \leq i \leq n$.
Moreover if $\sigma_1 \not= S$ then it is a $\tynint$- or a $\tyint$-type, and so, $\init(\sigma_1)$ by the IH, as desired.
\end{proof}

\begin{assumption}\label{ass:level-initial-limit-dz}
Recall signatures $\Sigma \subset \Sigma'$ from \cref{ass:sigma'}.
Henceforth we fix an level-initial limit \dz\ problem $\Gamma$ where \changed[lo]{$l = \max\{ \level(\rho) \mid X : \rho \in \Sigma' \setminus \Sigma\}$}.
\end{assumption}

\begin{definition}
\changed[lo]{Let $n \geq 1$.}
Given a $(\Sigma',\,\Hf)$-structure $\Bs$, define the \firstdef{\confused level-$n$ frame derived from $\Bs$}, written $\confusionf{\Bs}_n$, as follows.
\begin{compactitem}
\item $0 \leq \level(\sigma) \le n-2$, or $\sigma$ is not relational:
\lo{``not relational'' is not clear.
Surely we don't mean $\Z \to S$, which is not relational. I suppose the intension is to include $S$ (in case $n=1$).}
\[
\confusionf{\Bs}_n\sem{\sigma} := \Hf\sem{\sigma}
\]
\item $\level(\sigma)=n-1$:
\begin{align*}
\confusionf{\Bs}_n\sem{\Z\to\tynint} &:=
    \begin{array}{l}
    \big\{\Bs\sem{X} \, \overline{s} \; : \; X: \overline{\tyzero} \to \Z \to\tynint \in \Swouts,\\
    \qquad \qquad \qquad \quad |\overline s| = |\overline \tyzero|, s_i\in\Hf\sem{\tyzero_i}\big\}
    \end{array}\\
\confusionf{\Bs}_n\sem{\sigma_1\to\sigma_2} &:= \fspace{\confusionf{\Bs}_n\sem{\sigma_1}}{\confusionf{\Bs}_n\sem{\sigma_2}}
\end{align*}

\item $\level(\sigma)= n$:
\begin{align*}
\confusionf{\Bs}_n\sem{\Z\to\tynint}
&:=
\big\{f\in\fspace{\mathbb{Z}}{\confusionf{\Bs}_n\sem{\tynint}} : \forall i\le j .f(i)\le_\boolsort f(j) \big\}\\
\confusionf{\Bs}_n\sem{\tyzero\to\tyint}
&:= \fspace{\confusionf{\Bs}_n\sem{\tyzero}}{\confusionf{\Bs}_n\sem{\tyint}}
\end{align*}
\end{compactitem}

Note that $\confusionf{\Bs}_n$ does not have an interpretation of $\tyint$-types of level greater than $n$, nor $\tynint$-types of level greater than $n-1$.
\changed[lo]{In particular, $\confusionf{\Bs}_1$ is defined only for non-relational types, and $\tyint$-types of level 1, and $\tynint$-types and $\tyzero$-types of level 0.} \lo{CHECK!}
\end{definition}

We define $f\le_\boolsort g$ on booleans and relational functions by:
\begin{align*}
0 &\le_\boolsort 1 \\
f &\le_\boolsort g \quad \textrm{if $\forall x. f(x)\le_o g(x)$}
\end{align*}
{}
\lo{TODO. A simple example of $\Sigma \subset \Sigma'$ and structure $\Bs$ illustrating $\confusionf{\Bs}_2\sem{\sigma}$ for $\sigma = \Z \to o, (\Z \to o) \to o$ and $\Z \to (\Z \to o) \to o$.}

\begin{lemma}\label{lem:level-isframe}
Given a structure $(\Sigma',\,\Hf)$-structure $\Bs$,  $\confusionf{\Bs}_n$ is a frame.
\end{lemma}

\begin{proof}
We can use exactly the same proof of \cref{lem:isframe}.
\end{proof}

Let $\Sigma_i\supseteq\Sigma$ contain the predicate symbols of $\Swouts$ with types $\tynint$ such that $\level(\tynint)< i$ and types $\tyint$ such that $\level(\tyint)\le i$.




\begin{definition}
A family of structures \changed[lo]{$\{\Bs_n\}_{n \in \omega}$}, indexed by type level $n$, is said to be \firstdef{\confused} just if
$\Bs_0 = \As$, and
each $\Bs_{n+1}$ is a $(\Sigma_{n+1},\confusionf{\Bs_{n}}_{n+1})$-expansion of $\Bs_n$.

An \firstdef{entwined structure} is a member of an entwined family.
\end{definition}

\lo{From TCB:
The pattern goes: choose interpretations for predicates that are integer functions of level $n$; set the interpretation of level $n$ active types to be the finite set corresponding to the interpretations we've just picked; let the interpretation of level $n$ inactive types be the full function space between finite sets; let the interpretation of level $n+1$ active types be the monotone function space (which is countable); repeat.}

    \tcb{On consideration of the above remark, I believe we could allow $\tyint_1:=\tynint\to\tyint_2$ when $\level(\tynint)<\level(\tyint_2)$}

In the following lemmas, take \changed[lo]{$\{\Bs_n\}_{n \in \omega}$} to be an \confused family, and let $\Hf_n = \confusionf{\Bs_{n-1}}_n$.

\begin{lemma}
Let $\rho$ be a relational type. If $n>\level(\rho)$,
then $\Hf_n\sem{\rho}$ is finite.
\label{lem:finite}
\end{lemma}
\begin{proof}
We may assume that $\level(\rho)=n-1$, otherwise $\Hf_n\sem{\rho}=\Hf_{n-1}\sem{\rho}$.

If $\level(\rho)=0$ then $\Hf_n\sem{\rho} = \{0,1\}$ or a subset of a finite (by structural induction) function space, hence finite.

If $\level(\rho)>0$, then we work by induction on the structure of $\rho$. If $\rho=\Z\to \tynint$, then there are only finitely many $X\in\Swouts$; and since $\level(\tyzero_i)=0$, we have
$$\Hf_n\sem{\rho}=\{\Bs_{n-1}\sem{X} \, \overline s \mid
    X:\overline \tyzero \to\Z\to\tynint\in\Swouts,
                         s_i\in\Hf_n\sem{\tyzero_i}  \}$$
is finite. Otherwise, $\rho=\sigma_1\to\sigma_2$ and either $\sigma_1=S$ (which has a finite interpretation), or $\sigma_1$ has a relational type. By induction, $\bnsem{\sigma_1}$ and $\bnsem{\sigma_2}$ are finite, therefore so is $\bnsem{\rho}=\fspace{\bnsem{\sigma_1}}{\bnsem{\sigma_2}}$.
\end{proof}

\begin{lemma}\label{lem:lvnre}
If $n=\level(\tyint)$ then $\bnsem{\tyint}$ is recursively enumerable.

\changed[lo]{(N.B.~For any inactive type $\tynint$ with $\level(\tynint) = n \geq 1$, $\bnsem{\tynint}$ is not assumed to be defined.)}
\end{lemma}

\begin{proof}


By induction on \changed[lo]{$\textrm{arity}(\tyint)$}.
If $\tyint=\Z\to\tynint$ and 
\changed[lo]{$\level(\tyint)=n$}, then $\tynint=\sigma_1\to\cdots\to\sigma_k\to\boolsort$ and $\level(\sigma_i) \leq n-1$.
It follows from Lemma \ref{lem:finite} that each $\bnsem{\sigma_i}$ is finite.
Hence each monotone function $f : \fspace{\mathbb{Z}}{\Hf_n\sem{\tynint}}$ can be uniquely described by a map
\[
g :\fspace {\bnsem{\sigma_1}}{\cdots\to \fspace{\bnsem{\sigma_k}}{(\mathbb{Z}\cup \{\infty,-\infty\})}\cdots }
\]
where $f(x,\bar{s}) = [x \ge g(\bar{s})]$ (and $[\cdot]$ is the Iverson bracket).
The set of such $g$ is recursively enumerable.

Otherwise, $\tyint=\tyzero\to\tyint'$ with $\level(\tyzero) = 0$, so $\Hf_n\sem{\tyzero}$ is finite, hence $\fspace{\Hf_n\sem{\tyzero}}{\Hf_n\sem{\tyint'}}$ is r.e.~if $\Hf_n\sem{\tyint'}$ is, which holds by induction on arity.
\end{proof}

\begin{lemma}
The set of \confused families is recursively enumerable.
\end{lemma}
\begin{proof}
Since $\Swouts$ is finite, there is a maximum $m$ among the levels of the types of its symbols. For $i>m+1$, $\Bs_i$ is a structure extending $\Bs_{i-1}$ over the same alphabet, therefore $\Bs_i=\Bs_{m+1}$ by induction. Hence we only need to show that there are enumerably  many possible $\Bs_{m+1}$.

There is exactly one possible $\Bs_0$.

If we fix $\Bs_{n-1}$, then to construct the structure $\Bs_n$ we need to choose interpretations for
predicates in $\Swouts$ which are $\tyint$-type with level $n$, and predicates that are $\tynint$-type with level $n-1$.
\changed[lo]{(Notice that predicates that are $\tyint$-type with level $n-1$ already have an interpretation in the structure $\Bs_{n-1}$.)}
Since these choices must come from r.e. and finite sets respectively (by Lemmas~\ref{lem:lvnre} and~\ref{lem:finite}), there are enumerably many options for $\Bs_n$. Since an r.e.~collection of r.e.~sets is r.e., the number of families (up to level $m+1$) is r.e. by induction.
\end{proof}

\begin{lemma}\label{lem:emtonore}
If there is an \confused family such that $\Bs_m$ models $\Gamma$, then there is no resolution proof of $\bot$ from $\Gamma$.
\end{lemma}
\begin{proof}
We can use exactly the same proof of \cref{lem:order-emtonore}.

\end{proof}

\begin{lemma}\label{lem:smtoem}
If $\Gamma$ is standard-satisfiable then
 there is an \confused structure which models $\Gamma$.
\end{lemma}
\begin{proof}
%
%
%

We prove this using logical relations and Lemma~\ref{lem:termmonapp}.

As we construct an entwined family $\{\Bs_n\}_{n \in \omega}$,
we define logical relations \changed[lo]{$\arel^n$ between $\Hf_n := \confusionf{\Bs_{n-1}}_n$} and $\Sf$ recursively as follows:
\begin{align*}
j\arel_{\iota}^n i' & \defeq i=i'\\ 
b\arel_o^n b' & \defeq b\leq b'\\ 
r\arel_{\tau\to\sigma}^n r' & \defeq \forall s\in\Hf_n\tden\tau,s'\in\Sf\tden\tau \ldotp\\
& \qquad \qquad \qquad \qquad s\arel_\tau^n
s'\rightarrow r(s)\arel_\sigma^n r'(s')
\end{align*}
\changed[lo]{Thus $\arel^n$ is, by construction, the unique logical relation which is $\le$ on $\boolsort$ and is $=$ on other base types.}


\changed[lo]{Given a standard model $\Bs$, we construct a family of \confused structures by induction on $n$.
Let us suppose the $(\Sigma_{n-1},\confusionf{\Bs_{n-2}}_{n-1})$-structure $\Bs_{n-1}$ is defined.
We construct $\Bs_{n}$ as a $(\Sigma_{n},\confusionf{\Bs_{n-1}}_{n})$-expansion of $\Bs_{n-1}$ by:
\[
\Bs_n\sem{Y}f_1 \cdots f_k := \bigwedge_{\bar{f}\lr\bar{g}}\Bs\sem{Y}g_1 \cdots g_k
\]
where $Y:\sigma_1\to\sigma_2\to\cdots\to\sigma_k\to\boolsort$ is either of $\tyint$-type with level $n$, or of $\tynint$-type with level $n-1$, and each $f_i \in \Bs_n\sem{\sigma_i}$.}
\tcb{I just changed $\rho$ to $\sigma$ since it may be $S$}

\changed[lo]{We need to show that $\Bs_n$ is well-defined: $\Bs_n
\sem{Y} \in \Bs_n \sem{\overline \sigma \to o}$.
First suppose the former, i.e., $Y$ is of a $\tyint$-type
\[Y : \overline\tyzero \to \Z \to \overline \tynint \to \boolsort\]
with level $n$.\tcb{The types $\tyzero_i$ or $\tynint_i$ could also be $S$. It's possible that the easiest way to fix that is to include $S$ in these}
Then we need to show that $\Bs_n\sem{Y} \; \overline s$ is monotone, for each $s_i \in \Bs_n\sem{\tyzero_i}$.
I.e., assume $i \leq j$, we want to show that $\Bs_n\sem{Y} \; \overline s \; i \leq_o \Bs_n\sem{Y} \; \overline s \; j$;
or equivalently
\begin{align}
\forall f_1, \cdots, f_k \, . \,
\Bs_n\sem{Y} \; \overline s \; i \; f_1 \cdots f_k \leq \Bs_n\sem{Y} \; \overline s \; j \; f_1 \cdots f_k.
\label{eq:smtoen-monotone}
\end{align}
Now, since $\Bs$ is a model of $\Gamma$ which is assumed to be level-initial limit \dz, $\Bs\sem{Y} \; \overline t \in \Bs\sem{\Z \to \overline \tynint \to \boolsort}$ is upward-closed in the first argument for all $\overline s \lr \overline t$, i.e., for all $\overline f \lr \overline g$
\[
\Bs\sem{Y} \; \overline t \; i \; \overline g \leq \Bs\sem{Y} \; \overline t \; j \; \overline g
\]
which in turn implies~(\ref{eq:smtoen-monotone}), by considering the definition of $\Bs_n\sem{Y}$.}



\changed[lo]{Next suppose $Y$ is a $\tynint$-type with level $n-1$.
Given $f_i \in \Bs_n\sem{\sigma_i}$ for each $i$, we have $\Bs\sem{Y} \, \overline f = \bigwedge_{\overline f \lr \overline g} \Bs\sem{Y} \, \overline g$. [Prove this.] It follows from the definition that $\Bs_n\sem{Y} = \Bs\sem{Y} \in \Bs_n\sem{\overline\sigma \to o}$, as desired.}




Note that $\Bs_m\lr \Bs$. \lo{What is $m$? Clarify with reference to $l$ in \cref{ass:level-initial-limit-dz}.}
\tcb{As stated at \cref{point:hereiswhereImeanttofixm} (made a remark so that I could reference it now), $m$ is some integer greater than $l$. Taking $m=l+1$ is fine. (We need $l+1$ rather than $l$, because we want $\Bs\sem{\tynint}$ for inactive types $\tynint$ of level $m$ to be finite).}

\tcb{Could describe further ($\Bs_m\sem{X}$ is defined to be a lower bound for all argument combinations, matching the definition for $\lr$)} \lo{Yes please.}

We now show that $\Bs_m$ models $\Gamma$, by considering goal clauses and definite clauses separately. For goal clauses $\lnot G$, we show that for any $\alpha$, $\Bs_m\sem{G}(\alpha)=0$. For definite clauses $X\;\bar{x}\impliedby G\in \Gamma$, we show that for any $\alpha$, if $\Bs_m\sem{X\ \bar{x}}(\alpha)=0$ then $\Bs_m\sem{G}(\alpha)=0$.

Consider $\lnot G\in\Gamma$ where $G$ is positive existential. For any valuation $\alpha$ over $\Hf_m$, there is a standard valuation $\alpha'$ such that $\alpha'\glr\alpha$ (take $\alpha'(x)=\alpha(x)$ on variables of type $\iota$ or $S$ and take $\alpha'(x)=\top$ otherwise). Since $\Bs$ is a model, $\Bs\sem{G}(\alpha')=0$, and by Lemma \ref{lem:termmonapp}
, $\Bs_m\sem{G}(\alpha)=0$. Therefore $\lnot G$ is satisfied by $\Bs_m$.

Now consider $ X\;\bar{x}\impliedby G\in \Gamma$ and some $\Hf_m$-valuation $\alpha$. If $\Bs_m\sem{X\ \bar{x}}(\alpha)=0$ then $\Bs_m\sem{X}\alpha(x_1)\cdots \alpha(x_k)=0$. By construction of $\Bs_m$, there exist $g_1,\cdots, g_k$ such that $\Bs\sem{X}\overline g=0$ and each $g_i\glr\alpha(x_i) $. This allows us to
construct $\alpha'\glr \alpha$ such that $\Bs\sem{X\ \bar{x}}(\alpha')=0$. Again, since $\Bs$ is a model, $\Bs\sem{G}(\alpha)=0$, and by Lemma \ref{lem:termmonapp}
, $\Bs_m\sem{G}(\alpha)=0$. Therefore $X\;\bar{x}\impliedby G$ is satisfied by $\Bs_m$.

This shows that $\Bs_m$ is an entwined structure that models $\Gamma$.
\tcb{TODO: right now, Lemma~\ref{lem:termmonapp} only applies to applicative terms. The correctness of this argument relies on $b\arel_o b' \; \defeq \; b\leq b'$  --- that way it holds across $\land$ and $\lor$.}
\end{proof}

\begin{lemma}\label{lem:decidablemodel}
    Given an entwined structure $\Bs_m$, determining if it satisfies a goal or definite clause $G$ is decidable.
\end{lemma}
\begin{proof}
    \tcb{TODO: formalise}

    We prove this by converting the clause into a formula of linear integer arithmetic. The formulas at each step of the transformation are always terms, but some are neither HoCHC clauses nor first-order terms. \lo{First-order terms? One would expect formulas to be transformed to order-0 terms, and not first-order terms.}
    \tcb{``first-order terms'' is somewhat ambiguous. Here I mean formulas of LIA.}

    For each free variable $x$ in $G$ that is not of type $\Z$, $\Bs_m\sem{\Delta(x)}$ is finite.
    \lo{I take it that here $m$ is assumed to be greater the highest level of the predicates in $\Sigma' \setminus \Sigma$.}
    \tcb{Yes.}
    Therefore we may replace $G$ by a conjunction of clauses - one for each possible interpretation of $x$ (we introduce constants corresponding to each possibility).
    \lo{Strictly speaking, constant symbols are restricted to those in the first-order signature $\Sigma$, but it seems harmless to allow constants of a higher-order type.}
    \tcb{We could easily add them as new predicate symbols - just begin by extending $\Sigma'$ and $\Bs_m$}
    \lo{Note that this (para.) deals with the case of atoms of the form $x \; \overline M$ (in essence by reduction to atoms of the shape $X \; \overline N$).}

    We now describe how to turn applicative terms into formulas of linear integer arithmetic.
    For terms of the form $X\;x$,
    because \changed[lo]{$\Bs_m$} is an entwined 
    \changed[lo]{structure},
    $\Bs_m\sem{X\;x}$ is equivalent to $x>k_X$ (or $\mathit{true}$ or $\mathit{false}$).
    \lo{So we are assuming that $k_X \in \Z$ (as opposed to $\Z \cup \set{\infty, - \infty}$).}
    \tcb{I did this so that we end up with a formula strictly in LIA.}
    If the term is instead $X\;x\;M\;N$ where $M$ and $N$ are terms that do not contain free variables of type $\Z$, due to the preprocessing above, $M$ and $N$ have constant interpretations, so again we may replace the formula by $x>k_X\; M \; N$
    \lo{28/12: OK. By $k_X$ you mean an element of a finite set of functions.
    Equivalently we can take the formula translate to be $x>k_{X, M, N}$ where $k_{X, M, N}$ is an integer (depending on $X$ and ``constants'' $M$ and $N$).} (or $\mathit{true}$ or $\mathit{false}$).
    \lo{The argument works for terms of the form $X \; \overline M \; x \; \overline N$ where the $M_i$s and $N_j$s do not contain free variables of type $\Z$.}
    \tcb{Yes, although the $X\; \overline{M}\; x$ can be dealt with before thinking about $\overline N$. You end up with something like $\bigvee_{i}(x\in I_i\land f_i\; \overline{N}) )$ (where $f_i$ are constants as described in my next comment)}

    Otherwise, we have something like $X\; (Y\; y)\; x$ where both $x$ and $y$ have type $\Z$. In this case, by the level-initiality restriction, $Y\; y$ has an inactive type $\tynint$ (say),
    so $\Bs_m\sem{\tynint}$ is finite \changed[lo]{(because $\level(\tynint) < m$)}. Since $\Bs_m\sem{Y}$ is monotone in its first argument, $\Z$ can be divided into at most $|\Bs_m\sem{\tynint}| + 1$ intervals such that $\Bs_m\sem{Y}$ is constant within each interval.
    This means that we can break it into cases guarded by formulas asserting $y$ is in the appropriate interval.
    \lo{Thus, $\Bs_m\sem{Y}$ determines a subset $\set{b_1, \cdots, b_r} \subseteq \Bs_m\sem{\tynint}$, and a partition $\Z = \bigcup_{i=1}^r I_i$; and (using the notation in the preceding) we have that $\Bs_m\sem{X\; (Y \; y) \; x}$ is equivalent to $\bigwedge_{i=1}^r \big(y \in I_i \to  x > k_{X, b_i}\big)$.}

    \lo{28/12: More generally, suppose $X\; (Y \; y) \; x$ has a $\tynint$-type (which could be $\boolsort$, subsuming the preceding case) with $x, y$ of type $\Z$ as before, then
    \[
    \Bs_m\sem{X\; (Y \; y) \; x} = \mathsf{case}(x, y)\left[(x \in I_i) \wedge (y \in J_j) \mapsto b_{i, j}\right]_{i, j}
    \]
    where $\set{I_i}$ and $\set{J_j}$ are finite partitions of $\Z$.}
    \tcb{What you write is correct, but I believe my thoughts were slightly different and might make the recursion clearer. I would first transform $X\; (Y \; y) \; x$ into $\mathsf{case}(y)\left[(y \in J_j) \mapsto X\; f_j\; x\right]_j$ where $f_j$ are the new constants and then apply a transformation to each $X\; f_j\; x$.}

    \lo{28/12: Thus, take the term $X \; \overline M \; x \; \overline N : \tynint$ with $x : \Z$. Assume that the $\Z$-type variables in $\overline M$ and $\overline N$ are $x_1, \cdots, x_n$.
    Since $\Bs_m$ is an entwined structure, there exist:
    \begin{itemize}
        \item (corresponding to $x$) a finite partition $\set{I_k}_k$ of $\Z$, and
        \item for each $1 \leq i \leq n$, (corresponding to $x_i$) a finite partition $\set{I^i_{j_i}}_{j_i}$ of $\Z$,
        \item for each tuple $k, j_1, \cdots, j_n$ (ranging over a finite set), an element $b_{k, \overline j} \in \Bs_m\sem{\tynint}$
    \end{itemize} satisfying
    \[\begin{array}{rl}
    {}& \Bs_m\sem{X\; \overline M \; x \; \overline N}\\
    = & \mathsf{case}(x, \overline x)\left[(x \in I_k) \wedge \bigwedge_{i=1}^n(x_i \in I^i_{j_i}) \mapsto b_{k, \overline j}\right]_{k, \overline j}
    \end{array}\]
    (The case statement is easily expressible in (extended) LIA.)
    }
    \sr{29/12: Here $b_{k,\overline j} \in \Bs_m\sem{\tynint}$, but it is not clear to me how to express such objects in (first-order) LIA.  More difficult, I think, is how to determine the $b_{k,\overline j}$ from $X\ \overline{M}\ x\ \overline{N}$ effectively?
    \lo{31/12: My first response. I have completely reworked this proof - see the proof of the corresponding \cref{lem:order-decidablemodel} for the order-level-initial version of the problem. Two quick points. 1. The transformation of a term $M$ to a case-construct has been formalised as a big-step relation defined by induction over 3 rules, which should address your concern about effectivity.
    2. The claim is that if $M$ is a formula and $M \downarrow C$ then $C$ is expressible in LIA.
    Can you please see if the new proof addresses your concerns?}
    In particular, the $b_{k,\overline j}$ are a function of the interpretation of $\overline{M}$.  Take the simple case in which e.g. $\overline M$ is a single predicate symbol $Y$ of zero-arity.  This $Y$ may be defined by some definite clause $Y \Leftarrow G$ where $G$ is an arbitrary goal term.  Consequently, I don't see how determining the interpretation of $Y$ is any easier than the original problem.

    I have tried to sketch an alternative, but it at least needs that we restrict entwined structures to be monotone (which I guess is not a problem).

    The idea is to define all elements of an entwined preframe by a \emph{recursion-free} set of clauses over extended LIA.  Since extended LIA is decidable, it follows that resolution is complete for such sets of clauses because there will only be finitely many resolvents.

    We define a set of new relation symbols $R_d$, one for each element $d$ of the entwined preframe.
    Consider an arbitrary element (we only consider the integer function case):
    \[
        d \in \Bs_m\sem{\tyzero_1 \to \cdots{} \to \tyzero_k \to \Z \to \tynint_1 \to \cdots{} \to \tynint_m \to o}
    \]
    Each tuple of elements $(d_1,\ldots,d_k,e_1,\ldots,e_m)$ drawn from $\mathsf{D}(d) \coloneqq \Bs_m\sem{\tyzero_1} \times \cdots \times \Bs_m\sem{\tyzero_k} \times \Bs_m\sem{\tynint_1} \times \cdots{} \times \Bs_m\sem{\tynint_m}$ determines some $n \in \Z \cup \{\infty,-\infty\}$, computable from the enumeration, such that $d\ d_1\cdots{}d_k\ m\ e_1 \cdots{} e_m = 1$ iff $m \geq n$.  We write this uniquely determined $n$ as $\ulcorner (\overline{d},\overline{e}) \urcorner$.

    For each $d$ we construct the following clause:
    \[
            R_d\ \overline{x}\ y\ \overline{z} \Leftarrow \bigvee_{(\overline{d},\overline{e}) \in \mathsf{D}(d)} (\bigwedge_{i=1}^k \mathit{Is}_{d_i}(x_i) \wedge \bigwedge_{i=1}^m \mathit{Is}_{e_i}(z_i) \wedge y \geq \ulcorner (\overline{d},\overline{e}) \urcorner)
    \]
    The relation symbol $R_d$ acts as a definition for the element $d$ in the sense that it agrees with $d$ on finite arguments, but this will be quite complicated to state because $R_d$ should be interpreted in a frame built over extended LIA and $d$ is drawn from a frame built over LIA.
    The definition relies on a subsidiary family of formulas, abbreviated $\mathit{Is}$ with a subscript carrying the identity of a strictly lower-order element of the frame.
    \[
          \mathit{Is}_{d}(x) \coloneqq \bigwedge_{(\overline{d},\overline{e}) \in \mathsf{D}(d)} x\ \overline{R_{d_i}} \ \ulcorner (\overline{d},\overline{e}) \urcorner\ \overline{R_{e_i}}
    \]
    The formulas are true whenever their argument is no smaller, logically, than $d$, because they require that their argument holds whenever $d$ would hold.
    }

    Note that for any applicative term $M$ containing more than one integer variable, we can always find a subterm of the form $N\;x$ where $N$ does not contain any integer variable (and $x$ is an integer variable). We apply the case illustrated above.
\end{proof}

\section{Generalising the decidability argument}

\subsubsection*{Quasi-orders that contain no infinite antichains}

Recall that a \emph{well-quasi order} (WQO) is a quasi-order that is well-founded, and contains no infinite antichain.
It is well-known that in a WQO, every upward closed set has a finite representation (because ${\uparrow} (\min A) = A$, if $A$ is upward closed).

Less well-known is that every downward closed set in a WQO also has a finite representation:

\begin{theorem}[\citet{DBLP:conf/stacs/FinkelG09}]\label{thm:wqo-finite-rep-down-sets}
Every downward closed subset of a WQO is a finite union of ideals.\footnote{Let $(X, \leq)$ be a quasi-order.
An \emph{ideal} is a nonempty subset $I \subseteq X$ that is downward closed, and \emph{directed} i.e.~if $a, b \in I$ then there exists $c \in I$ such that $a, b \leq c$).}
\end{theorem}

Theorem~\ref{thm:wqo-finite-rep-down-sets} is in fact an immediate consequence of an old result of \citet{ErdosT43} which has been reproved many times.
See~\citet{DBLP:journals/lmcs/BlondinFM17} for a modern account (and the references therein).

\begin{theorem}[Erd\"os and Tarski]
A countable quasi-order $X$ contains no infinite antichain if, and only if, every downward closed subset of $X$ is equal to a finite union of ideals (moreover these ideals may be canonically chosen to be maximal ideals under inclusion).
\end{theorem}

However, if ideals are to serve as representatives (of downward closed sets), there must not be too many of them.
Alas, there are countable sets $X$ with an uncountable $\mathrm{Ideals}(X)$
(take e.g.~$X = \Sigma^\ast$ with the prefix ordering; then $\mathrm{Ideals}(X)$, isomorphic to $\Sigma^\ast \cup \Sigma^\omega$, is uncountable if $\Sigma$ has size at least 2).

Fortunately (\citet[Prop.~4.4]{DBLP:conf/icalp/BlondinFM14})
\begin{lemma}
A WQO $X$ is countable iff $\mathrm{Ideals}(X)$ is countable.
\end{lemma}

CHECK: This lemma extends to quasi-orders that contain no infinite antichains.

\paragraph{\dfrac{Assumption}{den} 1.}
\begin{enumerate}
\item $I$ is a countable quasi-order that contains no infinite antichain
\item $I$ has only countably many maximal ideals (all of which finitely representable)
\item $P$ consists of downward closed subsets of $I$.
\end{enumerate}

\begin{example}[Quasi-orders with no infinite antichains]~
\begin{enumerate}
\item The quasi-order of finite words over a finite alphabet, with lexicographical ordering, is not well-founded (take e.g.~the infinite descending sequence: $b, a \, b, a \, a \, b$, etc.), but contains no infinite antichains.
\item The set $\mathbb{Z}^n$ of $n$-tuples of integers, ordered lexicographically, is a quasi-order that contains no infinite antichains.
Because the ordering is total, if $I \subseteq \mathbb{Z}^n$ is downward closed, then $I \subseteq \mathrm{Ideals}(\mathbb{Z}^n)$.
\end{enumerate}
\end{example}

We seek examples of \emph{decidable} WQO and, more generally, quasi-orders that contain no infinite antichains.

The set of finite trees over a finite alphabet of tree constructors, ordered by tree embedding, is a WQO.
This is Kruskal's Tree Theorem \citet{Kruskal60}.
\citet{DBLP:conf/tapsoft/BoudetC93} showed that the positive existential fragment of the theory of tree embedding is decidable.
\else
\fi
\fi